\documentclass[
  reprint,
  aps,pra,
  superscriptaddress,
  nofootinbib,
  floatfix,
]{revtex4-2}

\usepackage{amsmath,amssymb,amsfonts,amsthm,mathtools,bm}
\usepackage{physics,braket,cancel}
\usepackage{graphicx}
\usepackage{dcolumn}
\usepackage{booktabs,multirow}
\usepackage[version=4]{mhchem}
\usepackage{siunitx}
\usepackage{tikz}
\usepackage{float}
\usepackage{xcolor}
\usepackage[colorlinks,allcolors=blue]{hyperref}
\usepackage{cleveref}
\usepackage{newtxtext}

\usepackage[most]{tcolorbox}

\newcommand{\bal}[1]{
\begin{align} #1 \end{align}
}

\newcommand{\baln}[1]{
\begin{align*} #1 \end{align*}
}

\def\lpar{\left(}
\def\rpar{\right)}
\def\lbra{\left[}
\def\rbra{\right]}

\def\circle{\scalebox{1.2}{\textbullet}}

\renewcommand{\braket}[1]{\left\langle #1 \right\rangle}
\renewcommand{\cref}[1]{(\ref{#1})}

\newcommand{\bea}{\begin{eqnarray}}
\newcommand{\eea}{\end{eqnarray}}
\newcommand{\beaa}{\begin{eqnarray*}}
\newcommand{\eeaa}{\end{eqnarray*}}
\newcommand{\del}[3]{\left1 3 \right2}
\newcommand{\avg}[1]{\left < #1 \right >}

\newcommand{\R}{\ensuremath{\mathbb{R}}}

\newtheorem{thrm}{Theorem}

\newtheorem{lemm}{Lemma}
\newtheorem{coro}[thrm]{Corollary}

\newcommand{\bi}{\begin{itemize}}
\newcommand{\ei}{\end{itemize}}
\newcommand{\bc}{\begin{center}}
\newcommand{\ec}{\end{center}}

\newcounter{mycomment}

\newcounter{example}
\newenvironment{example}[1][]{\refstepcounter{example}\par\medskip\noindent%
   \textbf{Example~\theexample.#1} \rmfamily}
   {\medskip\circle}

\newcounter{remark}
\newenvironment{remark}[1][]{\refstepcounter{remark}\par\medskip\noindent%
   \textbf{Remark~\theremark.#1} \rmfamily}{\medskip\circle}

\def\cC{ {\cal C} }

\def\cF{ {\cal F} }

\def\cD{ {\cal D} }

\def\cT{ {\cal T} }
\def\cU{ {\cal U} }

\def\cO{ {\cal O} }

\def\hc{\text{h.c.}}

\newcommand{\argmin}[1]{\underset{#1}{\text{argmin}}}

\def\integer{\mathbb{Z}}
\def\real{\mathbb{R}}
\def\complex{\mathbb{C}}
\def\goto{\rightarrow}
\def\epsilon{\varepsilon}

\newcommand{\cL}{{\cal L}}
\newcommand{\be}{\begin{equation}}
\newcommand{\ee}{\end{equation}}

\newcommand{\bigo}[1]{{\cal O}\left({#1}\right)}

\def\hc{\text{h.c.}}
\def\ntraj{N_{\text{traj}}}

\begin{document}

\title{Path integral control of open quantum systems}

\author{Aar\'on Villanueva}
\email{aaronv@science.ru.nl}
\affiliation{Radboud University, Heyendaalseweg 135, 6525 AJ Nijmegen, The Netherlands}

\author{Hilbert Kappen}
\affiliation{Radboud University, Heyendaalseweg 135, 6525 AJ Nijmegen, The Netherlands}

\begin{abstract}
We investigate open-loop quantum state preparation for a class of open quantum systems whose dynamics follow a Gorini–Kossakowski–Lindblad–Sudarshan (GKLS) master equation that admits a trajectory-based stochastic representation.
The deterministic control objective is reformulated as a stochastic optimal control problem---interpreting stochasticity as a methodological tool akin to stochastic Schrödinger equation unravelings---which situates the problem within the path integral control framework. For the class of GKLS generators under consideration, this reformulation leads to an explicit expression for the optimal control as a weighted average over stochastic quantum trajectories, thereby eliminating the need for gradient evaluations. Building on this theoretical result, we derive a control update rule for piecewise-constant control pulses and demonstrate that adaptive importance sampling progressively enhances the control estimator during optimization, culminating in the algorithm we term Path integral Quantum Control (PiQC). We further introduce an annealed variant of PiQC, wherein a synthetic noise schedule gradually steers open-system trajectories toward closed-system dynamics, enabling high-fidelity unitary state preparation. Numerical studies on a dissipative single-qubit system and a multi-qubit Nuclear Magnetic Resonance model verify that PiQC yields precise open-loop controls and displays robustness to Hamiltonian perturbations.
We propose PiQC as a trajectory-based alternative to gradient-based approaches, which might offer a viable solution in quantum control problems where gradient computation is infeasible or computationally demanding.
\end{abstract}

\maketitle

\section{Introduction}

A number of quantum optimal control techniques have been developed for open-loop control problems, i.e. problems in which the controls are time-dependent deterministic functions that are independent of the state~\cite{cong2014control,dAlessandro2021introduction}. In the unitary regime, diverse approaches have been proposed~\cite{mahesh2023quantum}, including the gradient-based Gradient Ascent Pulse Engineering (GRAPE) algorithm~\cite{khaneja2005optimal} and its variants~\cite{chen2023accelerating}, search-based methods such as Chopped Random-Basis Optimization (CRAB)~\cite{caneva2011chopped}, variational approaches like Krotov optimization~\cite{krotov1996global,tannor1992time,reich2012monotonically,Koch2019}, and techniques derived from the Pontryagin maximum principle (PMP)~\cite{boscainIntroductionPontryaginMaximum2021}. These methods have been successfully applied to a broad range of quantum control tasks; for reviews see~\cite{glaserTrainingSchrodingerCat2015,kochQuantumOptimalControl2022}.

\begin{figure*}[!t]
\centering
\includegraphics[width=0.6\textwidth]{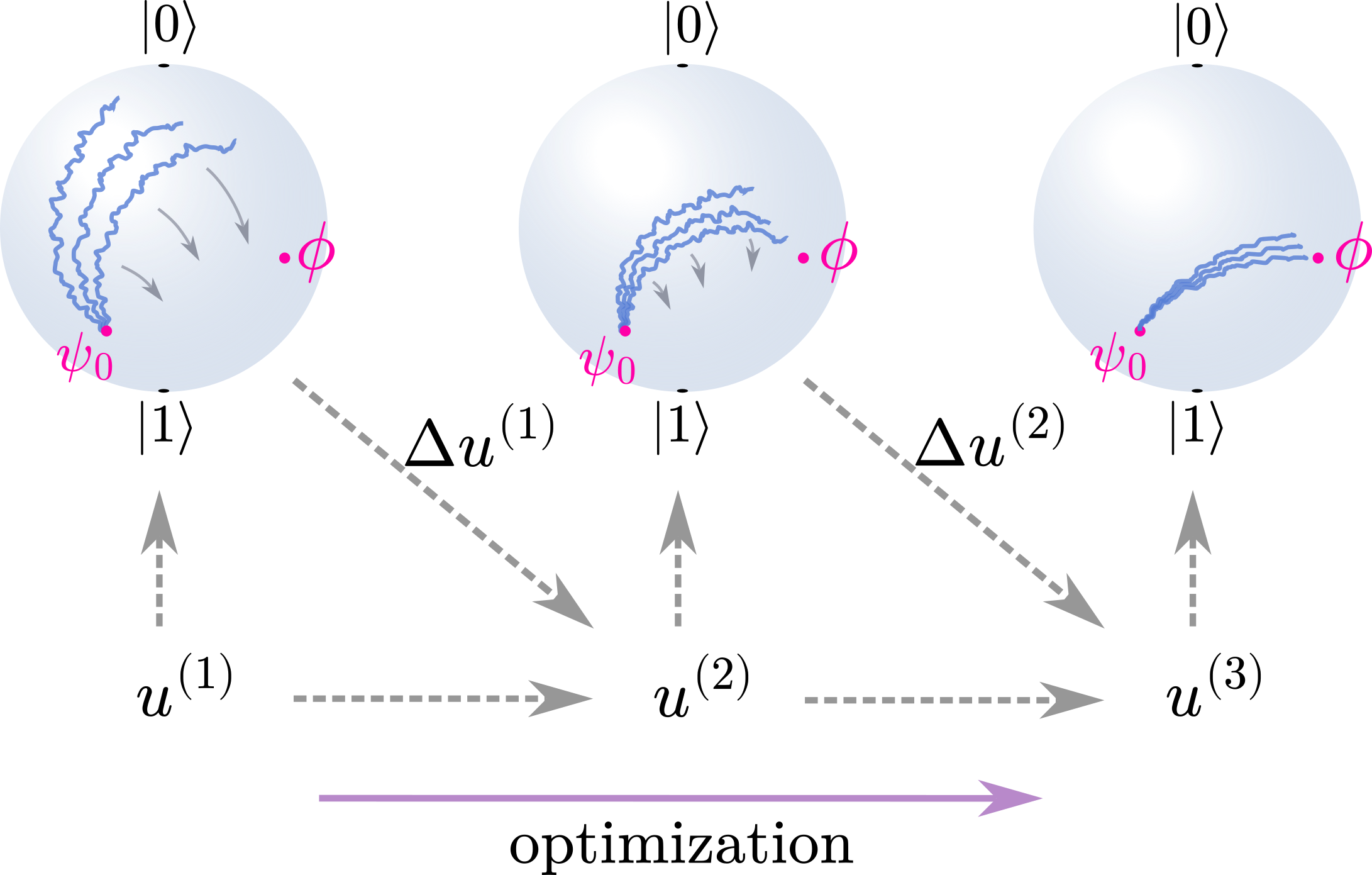}
\caption{Schematic representation of the workflow of our control algorithm for open-loop state preparation of a single qubit. Starting from an initial state $\psi_0$ and a target state $\phi$, a batch of quantum trajectories (blue solid lines) is generated under an initial control guess $u^{(1)}$. Adaptive importance sampling uses all trajectories in the batch to update $u^{(1)}$ to a new control $u^{(2)} = u^{(1)} + \Delta u^{(1)}$, such that trajectories in the next iteration move closer to the target. New trajectories are then sampled using $u^{(2)}$, and the procedure is repeated to obtain $u^{(3)},u^{(4)},\dots$. In the annealed variant of the algorithm (see Sec.~\ref{sec:qdcxclosed}), the noise strength in the trajectories is simultaneously reduced across iterations, which lowers trajectory variance and drives the average dynamics toward those of a closed system.}
\label{fig:bloch_traj}
\end{figure*}

In parallel, the control of open quantum systems has received substantial attention in recent years~\cite{Koch2016,kochQuantumOptimalControl2022}, albeit with less development than in the closed-system case.
Approaches based on direct simulation of master equations, such as the Gorini–Kossakowski–Lindblad–Sudarshan (GKLS) equation~\cite{GoriniKossakowskiSudarshan1976, Lindblad1976} (or, in short, Lindblad equation), become significantly more computationally demanding with increasing system size, since direct propagation of density matrices scales as $\bigo{N^2}-\bigo{N^3}$
in the Hilbert-space dimension.
In this context, the theory of stochastic \emph{quantum trajectories}~\cite{davies1969quantum,carmichael_open_1993,percival1998quantum,Daley2014} provides an alternative representation of open dynamics, in the presence or absence of measurements, and has found numerous applications~\cite{harrington2022engineered}.
Quantum trajectories can be viewed as noisy paths in the system’s Hilbert space that encode the interaction with the environment. In quantum control, such trajectories have been used primarily as a computational tool to approximate gradients in local optimization, benefiting from the quadratic scaling advantage of propagating wave functions of dimension $N=2^n$ instead of density matrices of dimension 
$N^2$~\cite{abdelhafez2019quantum,goerz2018efficient}.

A particularly important class of quantum trajectories are continuous diffusions described by stochastic Schr\"odinger equations (SSEs)~\cite{gisin1992quantum,percival1998quantum,barchielliQuantumTrajectoriesMeasurements2009}. From a control-theoretic perspective, SSEs suggest a correspondence between Lindblad dynamics and stochastic control problems: one may regard the deterministic open-loop state-preparation task as a control problem on the space of stochastic trajectories, where stochasticity is introduced as a methodological device in analogy with unravelings of the Lindblad equation, rather than as a physical feedback loop~\cite{wiseman1993quantum,belavkin1983theory,belavkin2013quantum,wiseman2009quantum}. In the fully general feedback setting, such stochastic optimal control (SOC) problems are typically approached via Hamilton–Jacobi–Bellman (HJB) equations or stochastic variants of the Pontryagin maximum principle formulated as backward stochastic differential equations, both of which are well known to be computationally demanding in high dimensions~\cite{pham2009continuous,yong_zhou99}. Notable exceptions include the classical linear–quadratic–Gaussian (LQG) setting, where the dynamics are linear, the running cost is quadratic in the control, and the noise enters additively. While SSEs with coherent controls are not of LQG type due to the bilinearity of terms such as $u\psi$, certain cases admit LQG formulations in the Heisenberg picture~\cite{wiseman2009quantum}, enabling rigorous treatments of measurement-based feedback control~\cite{doherty_feedback_1999,Tan2000,rossi_measurement-based_2018,magrini_real-time_2021}. In general, however, the control of open quantum systems governed by SSEs leads to nonlinear control problems for which efficient solution methods are scarce.

Within the broader SOC literature, the path integral (PI) control method~\cite{kap05a,thijssen2014a,kappen2015adaptive} was proposed to efficiently solve a large class of nonlinear problems in which the cost functional is quadratic in the control and may be nonlinear in the state, with LQG problems arising as a special case. A key advantage of the PI formulation is that the optimal control can be expressed in closed form as a path integral over dynamical trajectories in configuration space, allowing one to bypass the explicit solution of Bellman or PMP equations.
This representation in turn enables the use of well established statistical estimation techniques to approximate the optimal control.
An important example is adaptive importance sampling (IS)~\cite{kappen2015adaptive}, which iteratively refines a control estimator using self-generated trajectories driven by the current control, yielding a highly parallelizable method with modest computational overhead. PI control has demonstrated remarkable success in robotics, enabling solutions to high-dimensional, nonlinear control problems with real-time constraints where traditional methods often fail~\cite{williams2016aggressive}; see~\cite{kazim2024recent} for a recent review of PI control and related advances.

In this work we use the path integral (PI) control framework~\cite{kap05a,thijssen2014a,kappen2015adaptive} to address optimal open-loop control for a class of open quantum systems governed by the Lindblad equation.
We focus on state-preparation tasks where the dynamics admit a trajectory-based stochastic representation and reformulate the deterministic open-loop control objective as a stochastic optimal control problem on the space of quantum trajectories. In this setting, PI control yields an expression for the optimal control in terms of weighted averages over stochastic trajectories, which can be approximated by Monte Carlo sampling without computing gradients or solving Bellman equations.

Building on this connection, we derive a control update rule for parametrized open-loop controls, with an emphasis on piecewise-constant pulses. The update is obtained from global statistics of trajectory ensembles and implemented using adaptive importance sampling~\cite{kappen2015adaptive}, so that the sampling distribution and the control parameters are iteratively refined in tandem. A schematic illustration of this workflow for single-qubit state preparation is shown in Fig.~\ref{fig:bloch_traj}.
We consider controls that can be written as linear combinations of prescribed basis functions, which includes the common case of piecewise-constant pulse parametrizations. In this work, we use “pulse” to refer to the time-dependent profile of an external control field, parametrized as piecewise constant over a finite time grid.

Our use of quantum trajectories differs from previous work where unravelings are employed primarily to approximate gradients for deterministic optimal control governed by the Lindblad equation~\cite{boutin2017opengrape,goerz2018efficient, schmidt2011optimal}.
There, the trajectories serve as a numerical tool for gradient estimation.
We instead treat the trajectory description as the primary control space and formulate a stochastic control problem whose solution is approximated by adaptive importance sampling, leading to a gradient-free optimization scheme.
We refer to the resulting algorithm as Path integral Quantum Control (PiQC) and study it in the open-loop setting. We present two main variants: an “open-system” version applied directly to dissipative dynamics, and an “annealed” version in which a synthetic noise schedule interpolates between noisy and nearly unitary trajectories, yielding high-quality controls for unitary state preparation. Numerical examples include a dissipative single-qubit problem, where we compare to an open-system GRAPE implementation, and a multi-qubit Nuclear Magnetic Resonance (NMR) model used as a testbed for both open and nearly closed dynamics.

The remainder of the paper is organized as follows.
Section~\ref{sec:qcontrol} defines the quantum control problem for Lindblad dynamics.
Section~\ref{sec:pi_control} reviews the PI control formalism, defines the control model and notions of optimality, and introduces the adaptive importance sampling algorithm needed for the quantum control part.
Section~\ref{sec:piqc} maps the deterministic quantum control problem to a stochastic setting where PI control can be applied, and defines the PiQC algorithm for open systems, and a annealed variant for closed systems.
Section~\ref{sec:numerics} presents numerical experiments for single and multi-qubit systems, and Sec.~\ref{sec:discussion} discusses limitations and directions for future work.


\section{Quantum control problem definition}
\label{sec:qcontrol}

In what follows, quantum states are denoted by Greek letters such as $\psi$ and $\phi$, with their inner products written explicitly in the standard linear-algebraic form, e.g. $\psi^\dagger \phi$.

We define the state preparation problem as follows.
We aim to prepare a target quantum state starting from a given initial state in total time $T$, accounting for interactions between the system and its environment.
We assume the system dynamics follows the Lindblad equation~\cite{breuer2002theory}
\begin{gather}
\label{eq:lindblad}
\dot{\rho} = -i[H,\rho] + \cD[D', \cC'] (\rho)
\notag \\
\cD[D', \cC'] (\rho) := D'_{ab} \lpar C'_a\rho {C'_b}^\dagger -\frac{1}{2}\{ {C'_b}^\dagger C'_a, \rho \} \rpar
\end{gather}
where, from now on, we adopt the convention of repeated indices for summation.
The Hamiltonian $H$ takes the form $H(t)=H_0 + u_a(t)H_a$ with $H_a(a=1,\ldots, n_c)$ Hermitian operators (or control Hamiltonians), and $u_a(t)(a=1,\ldots, n_c)$ real-valued control fields.
The noise matrix $D$ is positive symmetric, and the dissipators $\cC' = \{C'_a\}_{a=1}^{n_c}$ satisfy a property to be determined.

Given the total time interval $[0,T]$, the cost objective is defined as
\bal{\label{eq:det_cost}
C_{\text{det}}[u]=-\frac{Q}{2}\cF(\rho_T)+ \frac{1}{2} \int_0^T  u(t)^T R u(t) dt
}
where $\rho_T$ is the density operator at final time $T$, $Q>0$ and $R$ is a real symmetric positive $n_c\times n_c$ matrix.
The first term corresponds to the overlap $\cF(\rho_T)=\Tr(\rho_T \phi\phi^\dagger)$ of the final state $\rho_T$ with respect to a target (pure) state $\phi$.
The second term, where $u(t)$ is the vector of controls $u_a(t)$, corresponds to the fluence, which penalizes the integrated amplitude of the control fields.
The task is to find the optimal open-loop control $u^*(t) \in \cU_o$, given a suitable open-loop control space $\cU_o$, that minimizes the objective function~\cref{eq:det_cost}, i.e.
\bal{\label{eq:det_prob}
u^* = \argmin{u \in \cU_o}\, C_{\text{det}}[u]
}

The Lindblad (GKLS) equation~\cref{eq:lindblad} and the cost~\cref{eq:det_cost} define a deterministic control problem.
It can be approached using, for example, the PMP method~\cite{boscainIntroductionPontryaginMaximum2021}.
Here, instead, we propose to use PI control.
Since PI control is a stochastic optimal control method, we first need to establish a suitable mapping between the deterministic setting given by \cref{eq:det_cost} and \cref{eq:det_prob}, and a stochastic control problem.
Before doing so, we formally introduce the PI theory and the tools needed to apply it to the quantum case.


\section{The path integral control method}
\label{sec:pi_control}

We review the path integral control theory~\cite{kap05a} which in its most general form yields the computation of optimal feedback control solutions.
For the presentation we follow Ref.~\cite{thijssen2014a}, which proves the main results behind the PI control method and importance sampling.
We apply these results to derive the equations needed for the importance sampling step used later in the text.
This section is meant to be self-contained, but the reader may consult~\cite{thijssen2014a} for further details.

\subsection{Introduction and mathematical definitions}

The path integral control framework is naturally expressed in the language of stochastic calculus using the Ito convention~\cite{oksendalStochasticDifferentialEquations2013}.
Consider a finite time interval $[t_0, T]$, where $t_0$ is the initial time and $T$ the final time.
Let $W_t$ denote an $n_c$-dimensional real-valued stochastic process with vanishing mean and variance proportional to the elapsed time $t - t_0$.
Such processes are referred to as Wiener processes, or equivalently Brownian diffusions~\cite{oksendalStochasticDifferentialEquations2013}.
Intuitively, one may regard a Wiener process as a temporal sequence of independent Gaussian increments $d W_t$, each of which with zero mean and variance proportional to $dt$.
In differential form, the Wiener increments satisfy $\braket{dW_t} = 0$ and $\braket{dW^a_t dW^b_t} = D_{ab} dt$, where $D$ is a time-independent real, symmetric, positive semidefinite matrix, and the expectation $\braket{\cdot}$ is taken over all sample paths of $W_t$ on $[t_0, T]$.
For compactness, we often write $\braket{dW_t dW_t^T} = D dt$, and refer to this as the Ito structure of the process.
When needed, we will write $\braket{\cdot}_{X_t = x}$ to denote expectation over trajectories conditioned on $X_t = x$ at time $t$.

Consider a dynamical system of the form
\bea \label{eq:PI1}
dX_t = f(t,X_t)dt + g(t,X_t)(u(t,X_t) dt + dW_t)
\eea
for $t_0 \le t \le T$ with $X_t$ an $M$-dimensional real stochastic variable, with initial condition $X_{t_0}=x_0$.
Here $W_t$ is a $n_c$-dimensional real Brownian motion with $\avg{dW^a_t}=0$ and $\avg{dW^a_t dW^b_t}=D_{ab} dt$ with $D$ a $n_c\times n_c$ positive symmetric matrix.
Define the spacetime domain $\Omega = [t_0,T]\times \R^M$.
Define the functions $f: \Omega \to \R^M$, $g: \Omega \to \R^{M\times n_c}$ and the feedback control $u : \Omega \goto \real^{n_c}$ with the sufficient regularity such that a solution to \cref{eq:PI1} exists.
The reader may consult~\cite{bensoussan2018estimation} for further details on these conditions.
Denote by $\cU$ the set of admissible feedback controls $u$.
Define the cost
\footnote{
Note that the expectation of the It\^o term $\int_t^T u(s,X_s)^T R dW_s$ vanishes because the control is adapted and the integral therefore defines a martingale~\cite{oksendalStochasticDifferentialEquations2013}.
However, we keep it in the definition for the sake of completeness and notational convenience, as this term appears in the definition of the importance sampling weights
when deriving the optimal solution of the path integral theory~\cite{thijssen2014a, kappen2015adaptive}.
}
\begin{widetext}
\bal{\label{eq:PI2}
S^u(t)= \Phi(X_{T})+ \int_t^{T} \left[ V(s,X_s) ds+\frac{1}{2} u(s, X_s)^T R u(s,X_s) ds  + u(s,X_s)^T R dW_s \right]
}
\end{widetext}
with $X_t$ the stochastic solution of \cref{eq:PI1} at time $t$.
The matrix $R$ is positive $n_c \times n_c$ symmetric such that $R=\lambda D^{-1}$ with $\lambda>0$ a constant.
The function $V$ is a state-dependent path cost, and $\Phi$ denotes the end cost.

The stochastic optimal control problem consists in finding the optimal feedback control $u^*(t, x) \in \cU$ that minimizes the expected cost
\bal{\label{eq:expected_cost}
C[u] = \braket{S^u(t_0)}
}
with initial condition $X_{t_0} = x_0$ at time $t_0$, in symbols
\bal{\label{eq:u*_def}
u^* = \argmin{u \in \cU}\, C[u]
}
We say that the stochastic differential equation \cref{eq:PI1} is in PI form or PI-compatible, and together with \cref{eq:expected_cost} and \cref{eq:u*_def} they define the PI control problem.
The generic approach to solve this would be to derive a partial differential equation for the so called cost-to-go, known as Bellman equation,
and derive the optimal control from the solution (see Lemma~\ref{lemma1} in the Supplementary Material).
However, for the class of PI-compatible control problems defined by Eqs. \cref{eq:PI1}, \cref{eq:PI2} and \cref{eq:u*_def}, one bypasses this requirement by establishing the following result.
\begin{thrm}
\label{thm1}
The optimal expected cost $C^* = C[u^*]$ and optimal control $u^*$ for the control problem \cref{eq:u*_def} with dynamics \cref{eq:PI1} and cost \cref{eq:expected_cost} is given by
\bal{\label{eq:opt_J}
C^* = -\lambda \log \avg{e^{-S^u(t_0)/\lambda}}
}
\bal{\label{eq:opt_u}
u^*(t,x)=u(t,x) + \lim_{dt \to 0}  \frac{\braket{\omega^u_t \Delta W_t}_{X_t = x}}{dt} \quad \forall (t, x) \in \Omega
}
$\omega^u_t := \frac{e^{-S^u(t)/\lambda}}{\braket{e^{-S^u(t)/\lambda}}_{X_t = x}}$, $u(t, x)\in \cU_t$ is an arbitrary control sampler, and $\Delta W_t = W_{t + \Delta t} - W_t$.
\end{thrm}
\begin{proof}
The proof is given in \cite{thijssen2014a} and reproduced in the Supplementary Material.
\end{proof}
Theorem~\ref{thm1} gives an explicit solution of the optimal expected cost $C^*$ and the optimal control for any $(t,x) \in \Omega$ in terms of path integrals, i.e. as weighted expectations over $u$-controlled trajectories traced in the spacetime $\Omega$.
This implies that the optimal control solution $u^*$ can be approximated using known sampling methods.
In this regard, we make an important remark.
\remark{
The control function $u$ in \cref{eq:opt_J} and \cref{eq:opt_u} is referred to as the sampling control or, in short, the sampler.
Note that the l.h.s. of \cref{eq:opt_J} and \cref{eq:opt_u} are {\it independent} of $u$: the same result holds for any sampler $u$, in particular, for $u(x,t)=0$.
However, naive sampling strategies can yield estimates with large statistical variance, for which the accuracy of one-step sampling estimates of $u^*$ heavily relies on a relatively good choice of the sampler $u$.
}
\bigskip

In practice, we must approximate the expectations in \cref{eq:opt_J} and \cref{eq:opt_u} using a time discretization and a finite number of samples.
Denote by $S^u_i(t)$ the cost of the $i$-th sample trajectory $X_{t:T}$ generated by a particular noise realization $W_{t:T}$.
Denote with $\ntraj$ the number of noise realizations/trajectories.
Each trajectory is weighted by $\omega^i_t=\frac{e^{-S^u_i(t)/\lambda}}{\frac{1}{\ntraj}\sum_{j=1}^{\ntraj} e^{ -S^u_j(t)/\lambda}}$.
When $S^u_i(t)$ has large variance, the batch of samples is dominated by one sample: the sample with lowest $S^u_i(t)$, in particular for low $\lambda$ (i.e. low noise).
The quality of the sampling can be quantified by the effective sample size~\cite{thijssen2014a}
\bea
ESS := N_\text{eff} / \ntraj \qquad  N_\text{eff} := \left(\sum_{i=1}^{\ntraj} (\omega_t^i)^2\right)^{-1}\qquad \label{eq:Neff}
\eea
where $1/\ntraj \le ESS \le 1$ and $ESS=1$ is maximal when $N_\text{eff}=\ntraj$ with all samples having equal weight $\omega_t^i=1/\ntraj$.
In~\cite{thijssen2014a}, it was shown that sampling controls with lower control cost~\cref{eq:PI2} are better samplers in the sense of smaller statistical variance and larger $ESS$.
It turns out that the optimal sampler $u$, in the sense of minimal variance, is given by the optimal control solution $u^*$, for which the variance of $S^u(t)$ is zero, and $\omega_t^i =1/\ntraj$, $ESS=1$.
We make an important remark.
\remark{
Any sampler $u$ induces an unbiased estimator for $C^*$ and $u^*$, but with different quality in terms of statistical variance.
When the sampler $u= u^*$, the variance of $S^u(t)$ becomes identically zero, which further implies that $\braket{S^{u^*}(t)} = S^{u^*}(t)$ and shows the self-consistency of \cref{eq:opt_u}, since $\braket{\omega^u_t \Delta W_t} = \omega^u_t \braket{ \Delta W_t} = 0$.
}
\bigskip

The remark above justifies the use of adaptive importance sampling for control estimation, as we discuss in Sec.~\ref{sec:IS}.

Finally, the quality of the control solution is in principle given by the total expected cost $\braket{S^u(t_0)}$.
Because of the monotonic relation between control cost and effective sample size~\cite{thijssen2014a}, one can also measure the quality of the control solution in terms of $ESS$ instead of $\braket{S^u(t_0)}$, which turns out to be a more sensitive measure of optimality.
See Sec.~\ref{sec:noisy-qubit}.


\subsection{Parametrized controls and optimality}
\label{sec:param_control}

The practicality of Equation~\cref{eq:opt_u} is limited to the computational complexity of estimating the optimal control $u^*(t, x)$ for each spacetime point $(t, x)\in \Omega$.
For high-dimensional problems, this task can turn challenging due to the curse of dimensionality.
Because of this, we must resort on a different strategy to approximate $u^*$.
One strategy is to assume that the control can be described by parametrized models in terms of known functions.
Here, we consider linear parametrizations in terms of basis functions.
Define a control model $\hat{u}$ as a linear parametrization of the form
\bal{\label{eq:linear_param}
\hat{u}_a(t,x) = \sum_{k\in I} \hat{A}_{ak} h_k(t, x)
}
for $a=1, \dots, n_c$, where $\hat{A}_{ak}$ are time-independent constants and $h_k:  \Omega \to \R$ are fixed basis functions.
The symbol $I$ denotes the index set for labeling basis functions.

For the moment, assume that the optimal control $u^*$ and the sampler control $u$ in \cref{eq:opt_u} belong to the Span of $\{h_k\}_{k\in I}$, i.e. $u^*$ and $u$ can be written in the form~\cref{eq:linear_param} as
\bal{\label{eq:linear_param_u_opt}
u^*_a(t,x) = \sum_{k\in I} A^*_{ak} h_k(t, x), \quad u_a(t,x) = \sum_{k\in I} A_{ak} h_k(t, x)
}
for $a=1, \dots, n_c$, and some $A^*_{ak},\, A_{ak} \in \real$.
Under this assumption, one can derive, using Theorem~\ref{thm2} from the Supplementary Material with $f(t,x)=h_{k'}(t,x)$, the following matrix equation
\bal{\label{eq:PI11}
\sum_k (A^*_{ak} - A_{ak}) B_{kk'} = C_{ak'}
}
with 
\bal{
B_{kk'}&=\avg{e^{-S^u(t_0)/\lambda}\int_{t_0}^T h_k(t,X_t)h_{k'}(t,X_t) dt} \notag\\
C_{ak} &=\avg{e^{-S^u(t_0)/\lambda}\int_{t_0}^T h_k (t,X_t)dW^a_t}\label{eq:PI11a}
}

Solving \cref{eq:PI11} for $\hat{A}$ in terms of the matrices $A, B$ and $C$ gives
\bal{ \label{eq:matrix_sol_1}
A^* = A + C B^{-1}
}
where we assumed $B$ to be invertible.

We remark that~\cref{eq:matrix_sol_1} can be derived as a particular case of Corollary 4 in \cite{thijssen2014a}.
Equation \cref{eq:matrix_sol_1} shows that when $u^*$ is of the parametrized form~\cref{eq:linear_param_u_opt}, $u^*$ can be approximated in terms of statistics $C, B$, which in turn can be computed once for a single batch of $\ntraj$ noise realizations $W_{t_0:T}$.
This is in stark contrast to the case of direct estimation of $u^*$, where we must solve for $u^*(t, x)$ for each $(t, x) \in \Omega$, which implies a sampling complexity $\sim \ntraj L$, with $L$ the total number of spacetime points in the discretization of $\Omega$.

What if $u^*$ does not belong to the Span of $\{h_k\}_{k\in I}$? In this case, one can prove that the solution given by \cref{eq:matrix_sol_1} is optimal in a precise sense.
A complete explanation of this case would take us farther than the scope of this work, therefore we will limit ourselves in briefly articulating the main ideas, while referring the reader to~\cite{kappen2015adaptive} for further details.
It turns out that the optimal control problem stated in \cref{eq:u*_def} can be formulated as a Kullback-Leibler divergence minimization problem between the distribution $p^*$ over state trajectories $X_{t_0:T}$ controlled by the optimal feedback control $u^*$, and the distribution $p^u$ over trajectories controlled by $u$~\cite{kappen2015adaptive}.
In symbols, we want to minimize $\text{KL}(p^* | p^u)$.

The problem of minimizing $\text{KL}(p^* | p^u)$ is convex in $p^u$, and the optimal solution is reached when $u = u^*$.
This corresponds to the unconstrained case in which the search space for $u$ is the whole $\cU$.
If the search space for $u$ is constrained, such as in the linear parametrization given by \cref{eq:linear_param}, then the solution \cref{eq:matrix_sol_1} is optimal not in the $\text{KL}(p^* | p^u)$ divergence sense, but in the flipped version $\text{KL}(p^u | p^*)$.
Obviously, minimizing the KL is not the same problem as minimizing its flipped version, but the two solutions approach each other when $u^*$ is sufficiently well approximated by a parametrized solution of the form \cref{eq:linear_param}, because the difference $\text{KL}( p^* | p^u) - \text{KL}(p^u | p^*) = \bigo{\delta p^3}$ with $\delta p = p^u - p^*$.
This is easy to see by Taylor-expanding the difference around $p^*$.
In analogy with classical mechanics, in this formulation the problem of finding the optimal control to the SOC problem can be seen as a variational problem where the KL divergence plays the role of an action.

The approach of minimizing the flipped version of the KL divergence is generally known in the literature as the Cross-Entropy (CE) method~\cite{de2005tutorial}, and it is commonly used as a proxy for approximating the solution to the original problem.
Along the text, we will adopt this point of view and frame the optimality of the solution given by \cref{eq:matrix_sol_1} within the context of minimizing the KL divergence $\text{KL}(p^u | p^*)$.

\subsection{Adaptive importance sampling and open-loop control}\label{sec:IS}

Using \cref{eq:matrix_sol_1} we can build an optimization scheme in which the control amplitudes $A$ are iteratively improved starting from an arbitrary guess.
One can iteratively apply~\cref{eq:matrix_sol_1} using the $A$ that is computed at iteration $p$ as the sampler for iteration $p+1$.
This is called adaptive importance sampling and the update step in the optimization takes the form
\bal{\label{eq:PI13}
A^{(p+1)}=A^{(p)} + C^{(p)} (B^{(p)})^{-1}
}
where $A^{(p)}, C^{(p)}, B^{(p)}$ are the statistics computed with $u^{(p)}$ specified by $A^{(p)}$ in the previous step.
We can initialize the algorithm with e.g. $A^{(0)}=0$, or with smarter initializations\footnote{For instance, with the solution of the corresponding deterministic control solution with $D=0$.}.

In this work, we are interested in open-loop control parametrizations, i.e. parametrizations of $u$ that are independent of the state $x$.
In particular, motivated by quantum control applications, we focus on piecewise-constant open-loop controls (often implemented as sequences of control pulses in experiments).
Consider a partition of the time interval $[t_0, T]$ in $K$ sub-intervals $I_k=[\tau_{k-1}, \tau_k]$ such that $t_0=\tau_0,\, \tau_1, \ldots, \tau_K=T$.
A piecewise constant open-loop parametrization can be defined by choosing the basis functions $h_k$ to adopt the form of indicator functions, $h_k(t) = \delta_{t\in I_k}$, where $\delta_{t\in I_k} = 1$ if $t\in I_k$, and zero otherwise.
The adaptive importance sampling equation \cref{eq:PI13} takes the form
\bal{\label{eq:open_loop_sol}
A^{(p+1)}_{ak}=A^{(p)}_{ak} + \avg{\omega_p \frac{\Delta W^a_k}{\Delta \tau_k}}_p
}
for $k=1,\ldots, K$, where
\bal{\label{eq:weights}
\omega_p := \frac{e^{-S^{u}(t_0)/\lambda}}{\avg{e^{-S^u(t_0)/\lambda}}_p}
}
and $\Delta W^a_k := \int_{I_k} dW^a_t$, $\Delta \tau_k := \tau_k - \tau_{k-1}$.
The average $\avg{\cdot }_p$ is computed with respect to trajectories generated by $u^{(p)}$ with initial condition $x_0$ at time $t_0$.

\section{Path integral Quantum Control}\label{sec:piqc}

\subsection{Stochastic mapping of the control problem}\label{sec:stoch}

Consider the controlled stochastic Schrödinger equation
\bal{\label{eq:control_linear}
d\psi_t = & -iH_0\psi_t dt -\frac{1}{2} D_{ab} H_aH_b \psi_t dt 
\notag \\ &-i H_a \psi_t (u_a(t) dt +dW_t^a)
}
where the $dW_t^a$ are real Wiener increments with zero mean and variance structure $\braket{dW_t dW^T_t} =D dt$, and $\braket{\cdot}$ denotes the average over all noise realizations.
The matrix $D$ is real symmetric and its relation with PI control and the quantum control problem is clarified below.
The open-loop control $u \in \cU_o$, and the Hermitian operators $H_a$, are defined as in \cref{eq:lindblad}.
Consider the expected cost
\par\vspace{-\baselineskip}%
\begin{widetext}
\bal{\label{eq:stoch_cost}
C[u]=\avg{-\frac{Q}{2}\cF(\psi_T)+ \frac{1}{2} \int_0^T u(t)^T R u(t) dt + \int_0^T u(t)^T R dW_t}
}
\end{widetext}
\vspace{-\baselineskip}\par%
where $\psi_t$ is the quantum state satisfying the SSE~\cref{eq:control_linear}, $\cF(\psi_t) :=\Tr(\psi_t\psi_t^\dagger \phi\phi^\dagger)$, and $\braket{\cdot}$ is the expectation over all quantum trajectories generated by \cref{eq:control_linear} with initial state $\psi_0$ at time $t=0$.
Note that the It\^o term $\int_0^T u(t)^T R dW_t$ vanishes under the expectation, but we include it here as we will need it for PI control---see discussion around~\cref{eq:PI2} in Sec.~\ref{sec:pi_control}.
Define the SOC problem
\bal{\label{eq:soc_prob}
u^* = \argmin{u \in \cU_o}\, C[u]
}

Finally, define the density operator $\rho(t) := \braket{\psi_t \psi_t^\dagger}$.
Then, the SOC problem given by \cref{eq:soc_prob} and the deterministic problem given by \cref{eq:det_prob} are equivalent in the sense that $C_\text{det}[u] = C[u]$ for all $u \in \cU_o$, provided that the dissipators $\cC'$ are related to the operators $C_a := -iH_a$ through the following linear transformation
\bal{\label{eq:invariant}
C_b=C'_a G_{ab}\qquad D'=G D G^\dagger 
}
where $G$ is an invertible matrix $G \in \complex^{n_c \times n_c}$, and $D'$ is the real symmetric appearing in the Lindblad equation~\cref{eq:lindblad}.
The proof follows directly by recognizing that~\cref{eq:control_linear} is an unraveling of the Lindblad equation with dissipation operator $\cD[D, \cC]$---this connection is well-established in the literature~\cite{barchielliQuantumTrajectoriesMeasurements2009} and we include a derivation in Appendix~\ref{app:unraveling} for completeness---combined with the fact that the dissipation operator remains unchanged under transformations of the form~\cref{eq:invariant}, i.e. $\cD[D, \cC] = \cD[D', \cC']$.

Note that the equivalence between the SOC formulation and the deterministic problem only holds in the case where $u$ is an open-loop control.
When $u$ is a feedback control, the mapping does not hold because control and state become correlated.

The stochastic mapping enables the use of path integral control techniques to address the quantum control problem by showing that~\cref{eq:control_linear} is PI-compatible; that is, it can be written in the form of \cref{eq:PI1}.~\footnote{
To see this, note that the dynamics in~\cref{eq:control_linear} can be written in the form
$$
d\psi_t =f(\psi_t) dt+g_a(\psi_t) (u_a(t) dt+ dW^a_t)
$$
with $f$ and $g_a$ complex valued vector functions.
Split the equation in real and imaginary components and define $X_t = (\psi_t^R, \psi_t^I)^T$, $F(\psi_t) = (f(\psi_t)^R, f(\psi_t)^I)^T$, $G_a(\psi_t) = (g_a(\psi_t)^R, g_a(\psi_t)^I)^T$.
Then, equation
$$
dX_t =F dt+G_a (u_a(t) dt+ dW_t^a)
$$
is PI-compatible.
}.
In summary, we can solve the deterministic problem \cref{eq:det_prob} using PI control provided that the condition $R = \lambda D^{-1}$ with $\lambda > 0$ is satisfied and the dissipators $\cC'$ are linearly related to anti-Hermitian operators through transformation~\cref{eq:invariant}.
Dissipation processes obeying this property arise naturally in practical settings, including emission and absorption channels $\cC' = \{\sigma^+, \sigma^-\}$ (see Secs.~\ref{sec:noisy-qubit} and \ref{sec:nmr} for representative examples), as well as dephasing noise with $\cC' = \{\sigma_z\}$; more broadly, they also include Pauli-type depolarizing noise, where $\cC'$ is drawn from the Pauli operator set.
All these scenarios can be mapped to anti-Hermitian dissipators.
A limitation, for example, arises for purely one-directional channels (e.g., zero-temperature amplitude damping or generalized amplitude damping~\cite{nielsenQuantumComputationQuantum2010}), which do not satisfy the condition without augmentation.
While the anti-Hermitian condition defines the range of applicability of the control framework developed here, in Section~\ref{sec:discussion} we discuss possible pathways to circumvent such a limitation.


\subsection{Algorithm definition for open systems}

To finalize the definition of the diffusion algorithm, we must specify a suitable open-loop search space $\cU_o$.
As for quantum control applications, we are interested in pulse-based control.
To this end, we choose $\cU_o$ to be the space of piecewise constant functions as defined in Sec.~\ref{sec:IS}.
We refer to this algorithm as Path integral Quantum Control and summarize it as follows.
Set an initial sampler $u$.

\begin{enumerate}
\item Sample $\ntraj$ noise trajectories $W_{0:T}$
\item Generate $\ntraj$ state trajectories $\psi_{0:T}$ using \cref{eq:control_linear} and record the final states $\psi_T$
\item Compute $\ntraj$ weights $\omega$ using \cref{eq:weights}
\item Update $u$ using \cref{eq:open_loop_sol}
\item Repeat
\end{enumerate}

In practice, due to the statistical nature of the algorithm, the controls can exhibit high variance.
To reduce this variance and improve convergence, we add a smoothing procedure in Step~4.
This is done as follows.
In Step~4, instead of estimating the next control $u^{(p+1)}$ using $u^{(p)}$, use a recent average of controls $u^{(i)}$ computed over a window of past importance sampling steps.
Operationally, this means replacing $A_{ak}^p$ on the right hand side of~\cref{eq:open_loop_sol} (in both terms) by the average $\frac{1}{w}\sum_{p'=p-w+1}^p A_{ak}^{p'}$, where $w$ is a user-defined window size.
This smoothing ansatz can improve the quality of solutions at the cost of slowing down the convergence.


\subsection{Annealed variant for closed systems}\label{sec:qdcxclosed}

We propose a variant of the {PiQC} algorithm for approximating optimal controls for closed systems.
We emphasize that the annealed variant of PiQC introduced in this section is proposed as a heuristic gradient-free algorithm, 
and no claim of optimality is made for the closed-system case.
For settings in which gradients are available, established methods such as GRAPE, Krotov, or PMP-based approaches remain standard alternatives and may be preferable.
The numerical evidence presented in Section~\ref{sec:nmr} demonstrates that the heuristic performs well in the benchmarks considered, but a rigorous convergence 
analysis in the unitary limit is beyond the scope of this work and is left for future investigation.

Consider the deterministic control problem with dynamics given by the Liouville equation resulting from setting the dissipation term in \cref{eq:lindblad} to zero.
This can be achieved e.g. by setting the noise matrix $D$ to zero.
More explicitly, consider the closed dynamics
\bal{\label{eq:liouville}
\dot{\rho} = -i[H_0 + u_a H_a, \rho]
}
together with the cost objective \cref{eq:det_cost}.
If the initial state $\rho_0 = \psi_0 \psi^\dagger_0$, \cref{eq:liouville} describes a purely unitary evolution.

\begin{figure*}[!t]
    \begin{center}
\includegraphics[width=0.6\textwidth]{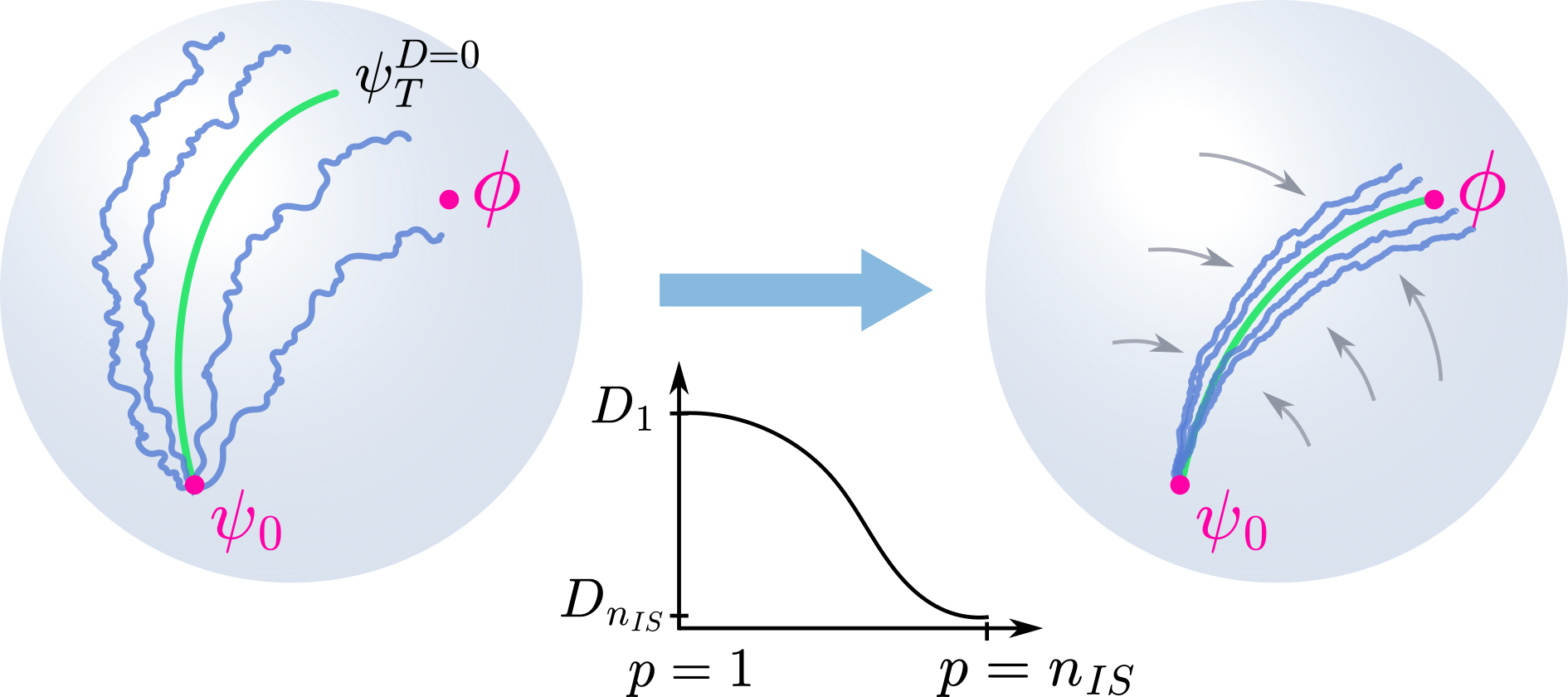}
    \end{center}
    \caption{Schematic representation of the annealed diffusion control algorithm process.
    In early stages of the optimization, the amplitudes of the noise matrix $D_p$ are relatively high and the trajectories (solid blue lines) present high variance with respect to the solution corresponding to $D=0$ (green solid line).
    By the end of the optimization, the noise amplitudes are very small, and the trajectories tend to align with the trajectory corresponding to $D=0$ due to a reduction of variance, while reaching the target $\phi$.
}
    \label{fig:bloch_anneal}
\end{figure*}

We can approach the problem of minimizing the cost objective above given dynamics \cref{eq:liouville} using the PiQC algorithm by modifying the noise matrix $D$ across optimization steps, such that $D \goto 0$ at the end of the optimization.
Intuitively, this means approaching the closed dynamics by progressively switching off the interactions with the environment.

More concretely, consider a series of noise matrices
\baln{
D_1, D_2, \ldots, D_{n_{IS}}
}
with $n_{IS}$ the number of importance sampling steps.
The noise matrices are different instances of the noise matrix that defines the variance structure in Eq.~\cref{eq:control_linear}.
Recall that, for the PI control method to be applicable, the relation $R = \lambda D^{-1}$ should be satisfied.
Then, $D_{n_{IS}}$ cannot vanish.
In consequence, we will choose $D_{n_{IS}}$ such that $\|D_{n_{IS}}\| = \epsilon$ with $\epsilon$ sufficiently small but not zero.
The set $\{D_p\}$ will start from relatively high noise strength values at $p=1$ and will gradually reduce such that $\|D_{n_{IS}}\| = \epsilon$.
We denote such a set $\{D_p\}$ as a control annealing schedule or simply, annealing schedule.

Another important point is that, throughout the annealing, the matrix $R$ is held fixed and the parameter $\lambda$ is updated at each step as $\lambda_p = R D_p$, consistent with the PI condition $R = \lambda D^{-1}$.
The annealing schedule is thus fully specified by the sequence $\{D_p\}$ alone.

We note that this annealing procedure is distinct from quantum annealing methods used for combinatorial optimization; here, annealing refers to a schedule on synthetic noise strength.
We use the term {\it annealing} drawing from the fact that the progression from high-noise values at the beginning of the optimization towards low-noise values at the end of it resembles the process of {\it cooling down} in classical annealing.
At $p=1$, we start with very noisy quantum trajectories in the Hilbert space induced by high values of the noise strength $D_1$.
As the optimization progresses, these trajectories become less noisy, decreasing their relative variance, due to a weaker interaction with the environment induced by a smaller noise strength $D_p$.
At the end of the optimization at $p=n_{IS}$, the quantum trajectories tend to coalesce into a single trajectory dictated by the unitary dynamics \cref{eq:liouville}.
In this sense, we {\it cool down} the stochastic trajectories to finally merge into a single unitary evolution.
We remark that the analogy with classical annealing is not related to the temperature of the device, nor to any type of physical temperature.
A picture of temperature arises only from the use of synthetic noise induced by the schedule $D_p$, which controls the amplitude of the noise manifested in the quantum trajectories.

If a good enough schedule $\{D_p\}$ is chosen, the algorithm will anneal towards a solution $u^{(n_{IS})}$ that is a good approximation to the true solution of the unitary problem given by \cref{eq:liouville} and \cref{eq:det_cost}.
In this variation, PiQC can be regarded as a quantum‑control analog of classical annealing methods.
We illustrate this process in Fig.~\ref{fig:bloch_anneal}.


\section{Numerical experiments}
\label{sec:numerics}

In this section we apply the PiQC algorithm to several examples to showcase its effectiveness.
In what follows, we assume open-loop control parametrizations in the form of piecewise constant control pulses over $K$ time intervals.
We parametrize the controls according to~\cref{eq:linear_param}, and write
\bal{\label{eq:openloop}
u_a(t) = \sum_{k=1}^{K} u_{ak} \delta_{t\in I_k} \quad a=1, \ldots, n_c
}
with $I_k = [\tau_{k-1}, \tau_k]$, $u_{ak}$ the $k$-th pulse corresponding to control $u_a$, and $n_c$ the total number of controls.
To simulate the stochastic dynamics~\cref{eq:control_linear}, we use the Euler-Maruyama integration scheme.
We made the code implementation supporting this section publicly available at~\cite{github_piqc}.


\subsection{Control of a noisy qubit} \label{sec:noisy-qubit}
Consider a single-qubit system evolving according to the Lindblad equation \cref{eq:lindblad} with $H=u_x H_x + u_y H_y$, where $H_x = \sigma_x$ and $H_y=\sigma_y$.
The dissipation part is given by the two non-Hermitian operators $C'_1 =\sigma^+$ and $C'_2 =\sigma^-$, where $\sigma^\pm = \frac{1}{2}(\sigma_x \pm i \sigma_y)$.
This system is commonly used for modeling the emission and absorption of light quanta in a two-level system coupled to an electromagnetic field, e.g. a cavity resonator~\cite{Daley2014}.
Assume a diagonal noise matrix $D'_{ab}=D \delta_{ab}$ with $D$ real positive.
The Lindblad equation is
\bal{\label{eq:lindblad_1qubit}
\dot \rho =-i[H, \rho]  + D  \lpar \sigma^+ \rho \sigma^- + \sigma^- \rho \sigma^+ -\rho \rpar
}
To propose an unraveling suitable for PI formulation, we  transform the dissipators to a pair of anti-Hermitian operators using~\cref{eq:invariant} with $C_1=-i H_x$, $C_2=-iH_y$,
\bal{\label{eq:one-qubit-invariance}
G = 
\begin{pmatrix}
-i &  -1 \\
-i & 1
\end{pmatrix} 
}
The unraveling~\cref{eq:control_linear} becomes
\bal{\label{eq:one-qubit-psi}
d\psi_t&=-D\psi_t dt -i \sum_{a=x, y} \sigma_a \psi_t (u_a(t) dt +dW_t^a) 
}
As stated in Sec.~\ref{sec:stoch}, unraveling \cref{eq:one-qubit-psi} is norm-preserving, which implies that all the generated stochastic trajectories will lie on the Bloch sphere.

\begin{figure*}[!t]
    \begin{center}
\includegraphics[width=0.8\textwidth]{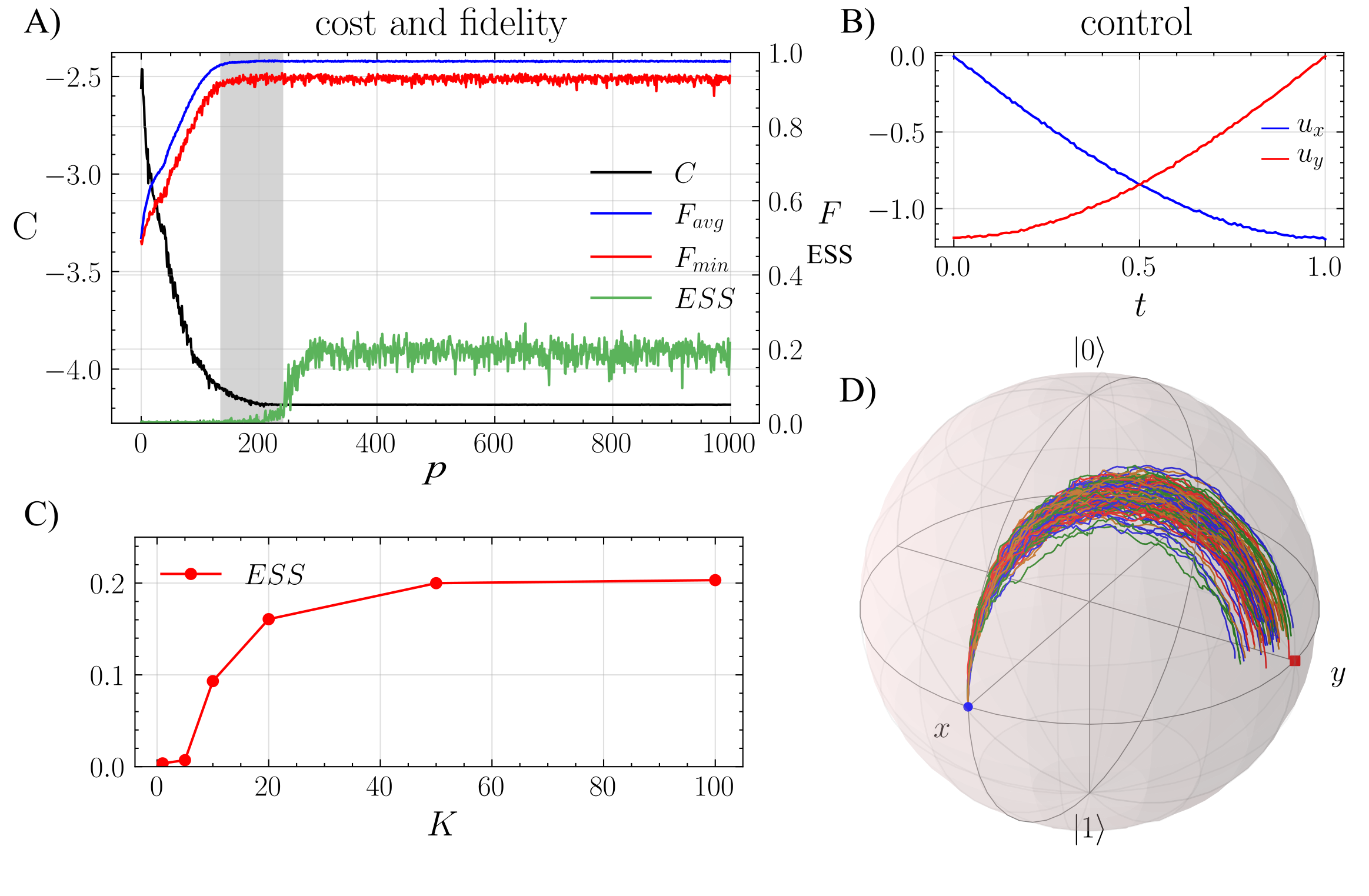}
    \end{center}
    \caption{
   Control of a noisy qubit from $X\to Y$. ({\bf A}) Average fidelity $F_\text{avg}$, worst-case fidelity $F_\text{min}$, effective sample size ESS and cost $C$ in~\ref{eq:stoch_cost} versus importance sampling iterations $p$.
    The converged average fidelity is $F^* = 0.9759 \pm 0.0006$.
   ({\bf B}) Optimal control solution $u_{x,y}(t)$ after convergence of the algorithm.
   The optimal solution is two-fold degenerate. The two solutions are related by a global sign $u_{x, y} \goto -u_{x, y}$.
({\bf C}). Dependence of the quality of the optimal control, measured by the ESS, on the number of pulses $K$.
({\bf D}). Optimally controlled trajectories on the Bloch sphere.
Parameters: final time $T=1$, noise coupling $D=0.005$, control weight $R= 1$, fidelity weight $Q=10$, number of pulses $K=128$, number of trajectories $\ntraj=400$ per IS step, maximum number of IS steps $n_{IS}=1000$ and time discretization $dt=T/N_T$ with $N_T=128$. IS smoothing window $w=40$ (see Section~\ref{sec:IS}).
}
    \label{fig:onequbit_XtoY}
\end{figure*}

We use as cost Eq.~\cref{eq:stoch_cost}, subject to $R D = \lambda$.
We consider the state preparation problem from an initial state $\psi_0=\ket{X}:=\frac{1}{\sqrt{2}}(\ket{0} + \ket{1})$ to 
a target state $\phi = \ket{Y} := \frac{1}{\sqrt{2}}(\ket{0} + i \ket{1})$.
Although the transformation $X\to Y$ can be realized by a simple $\sigma_z$ rotation, this task is non-trivial, since $\sigma_z$ is not one of the control primitives in $H$.
Instead, $\sigma_x,\sigma_y$ must coordinate so as to realize the desired rotation.
As a result, the optimal trajectories do not lie on the equator.

In Fig.~\ref{fig:onequbit_XtoY} (A) we show the adaptive importance sampling.
We plot the average fidelity over trajectories $F_\text{avg}$, the worst-case fidelity over trajectories $F_\text{min}$ and the average control cost $C$ in~\cref{eq:stoch_cost} versus IS steps.
In addition, we plot the effective sample size $ESS$ in~\cref{eq:Neff}, which is a sensitive measure of the quality of the optimal control solution.
The ESS indicates how close the control solution is to the optimal (feedback) control, for which $ESS=1$.
We observe that while the fidelity and control cost converge fast to constant values, the $ESS$ still increases indicating that the quality of the control solution is still improving (shaded region).
We observe that the control solution becomes smoother in these later IS iterations.
The asymptotic average fidelity for $K=128$ is $F^* = 0.9759 \pm 0.0006$.
Since the optimal control solution is a compromise between the final fidelity and the fluence $U := \int_0^T u^T(t) u(t) dt$, higher fidelity solutions can be obtained by lowering $R$ in Eq.~\cref{eq:stoch_cost}.
Indeed, by reducing $R=1$ to $R=0.1$ the average asymptotic fidelity increases to $F^* = 0.9963 \pm 0.0002$.
Fig.~\ref{fig:onequbit_XtoY} (B) shows one of the two optimal control solutions $u_{x,y}(t)$ after convergence of the IS training.
The other solution is obtained by taking $u_{x,y}(t)\to -u_{x,y}(t)$.
In Fig.~\ref{fig:onequbit_XtoY} (C) we show how the quality of the optimal control solution in terms of asymptotic $ESS$ depends on the number of pulses $K$. 
The $ESS$ increases monotonically until reaching an asymptote at around $ESS \sim 0.21$, showing the sub-optimality of the open-loop control compared to the optimal feedback solution (for which $ESS=1$).
Fig.~\ref{fig:onequbit_XtoY} (D) shows quantum trajectories on the Bloch sphere under optimal control.

\noindent {\bf Comparison with Open GRAPE}.---
We compare the quality of solutions between PiQC and (first order) Open GRAPE algorithm~\cite{boutin2017opengrape}, which compute coherent controls for the Lindblad equation.
See the Supplementary Material for details on the algorithm and implementation.
We solve the state preparation problem $X \goto Y$ for the single qubit using $K=128$ pulses.
The control problem is defined by the Lindblad equation~\cref{eq:lindblad_1qubit} and the cost function~\cref{eq:det_cost}.

\begin{figure*}[!t]
\begin{center}
\includegraphics[width=0.9\textwidth]{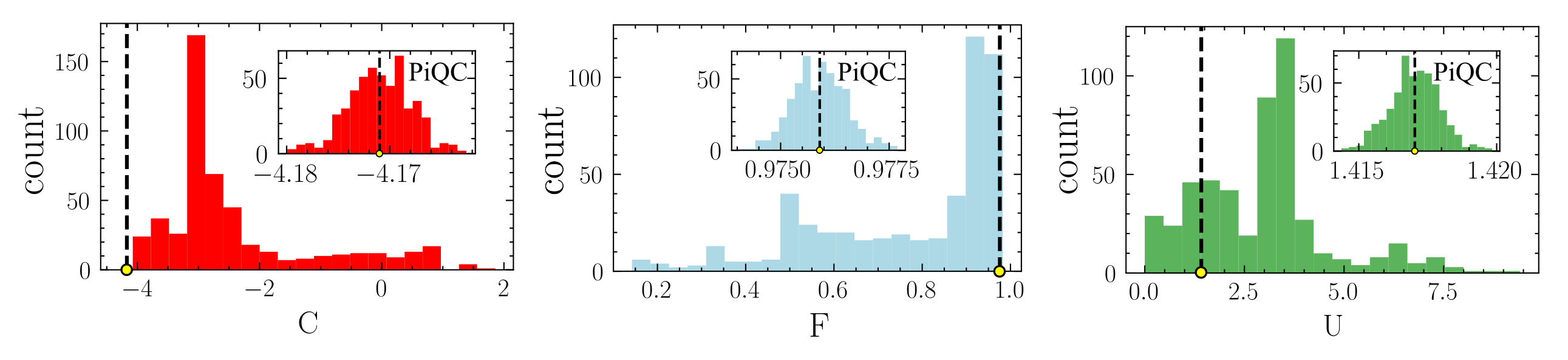}
\hspace{0.5cm}
    \end{center}
    \caption{Histograms of optimal control costs obtained with Open GRAPE (main panel) and PiQC (inset) on $505$ runs.
    For Open GRAPE, the average cost, fidelity and fluence (and standard deviations) are $\braket{C} = -4 \pm 1$, $\braket{F} = 0.8 \pm 0.2$ and $\braket{U} = 3 \pm 2$, respectively.
    The minimum cost achieved is $C_{min} = -4.07$.
For PiQC, the corresponding average values (yellow dot) are $\braket{C} = -4.171 \pm 0.003$, $\braket{F} = 0.9759 \pm 0.0006$ and $\braket{U} = 1.4170 \pm 0.0009$.
In the main panels, the dashed vertical lines at the PiQC average are wider than the actual standard deviations of the PiQC distributions.
This illustrates the large difference in variance between the two methods.
All PiQC runs converge to solutions with the same cost, fidelity, fluence and effective sample size up to small statistical fluctuations due to the stochastic nature of PiQC algorithm.
Parameter setting: $D=0.005$, $T=1$, $R=1$, $Q=10, K=128$, $w=20$.
}
    \label{fig:grape_hist}
\end{figure*}

The convergence of GRAPE-type algorithms depend on the specific problem and initial guesses, therefore it is common practice to run the algorithm with multiple random control pulses and select those that are successful~\cite{chen2025robust}.
We executed the PiQC and Open GRAPE algorithms 505 times, utilizing random control pulse initializations sampled from a normal distribution (centered at 1, 0, and -1, with standard deviation 2), and from a uniform distribution (integers between -10 and 10, with a scale factor of 1 and 0.5).
Regarding the stopping condition for Open GRAPE, the optimization stops when the cost cannot be further optimized by decreasing the learning rate; see App.~\ref{app:opengrape} for details.
On the other hand, PiQC converges within 800 IS steps utilizing a single control initialization, i.e. $u^{(0)} = 0$.

In Fig.~\ref{fig:grape_hist} we show the distribution of control costs (along the fidelities and fluences) for Open GRAPE  (main panels) and PiQC (insets).
We mark the average cost, fidelity and fluence of PiQC solutions by a yellow dot in each graph.
The PiQC solutions have lower cost than the Open GRAPE solutions.
The width of the dashed vertical lines locating each mean is much larger than the standard deviation of each PiQC distribution, showing that the variance of PiQC solutions is very narrow compared to that of Open GRAPE ones.
Open GRAPE converges to many different solutions with different costs depending on the initial condition.
Instead, PiQC converges consistently to one of two solutions with same cost, fidelity, fluence and effective sample size, which are related by a global sign (Fig.~\ref{fig:onequbit_XtoY} (B)).
The difference between individual solutions has root-mean-square (RMS) error less than $\bigo{10^{-4}}$.
By direct inspection we can also verify that these profiles correspond to two different underlying smooth solutions.
We conclude that the different PiQC solutions, rather than representing local minima, correspond to the same underlying solution up to small statistical fluctuations due to the stochastic nature of the algorithm.

We further compare these two algorithms' solutions by performing a t-distributed stochastic neighbor embedding (t-SNE) embedding~\cite{van2008visualizing} over Open GRAPE solutions together with PiQC solutions.
In Fig.~\ref{fig:tsne} (A) we plot some of the most regular Open GRAPE control solutions.
In Fig.~\ref{fig:tsne} (B) we show the t-SNE embeddings for $[u_x, u_y]$, considered as a vector of length $2K$.
In Fig.~\ref{fig:tsne} (C) we plot the average profiles of the Open GRAPE solutions of cluster 1 (red) and cluster 2 (blue) together with the average PiQC profiles.
Each dot represent a different Open GRAPE solution and we represent the (average) PiQC solutions with yellow dots.
We further perform a K-means decomposition to cluster the different solutions into 3 clusters (colored in red, blue and black).
The different Open GRAPE solutions tend to cluster around each ground truth PiQC solution.
The qualitative resemblance between the averages and the PiQC solution indicates that the Open GRAPE local minima tend to cluster around the ground truth solutions represented by PiQC solutions.
The cluster colored with black consists of local minima control solution of Open GRAPE with high control cost $C$.
We excluded the black cluster from the analysis because they correspond to high cost solutions.
\begin{figure*}[!t]
\begin{center}
\includegraphics[width=0.7\textwidth]{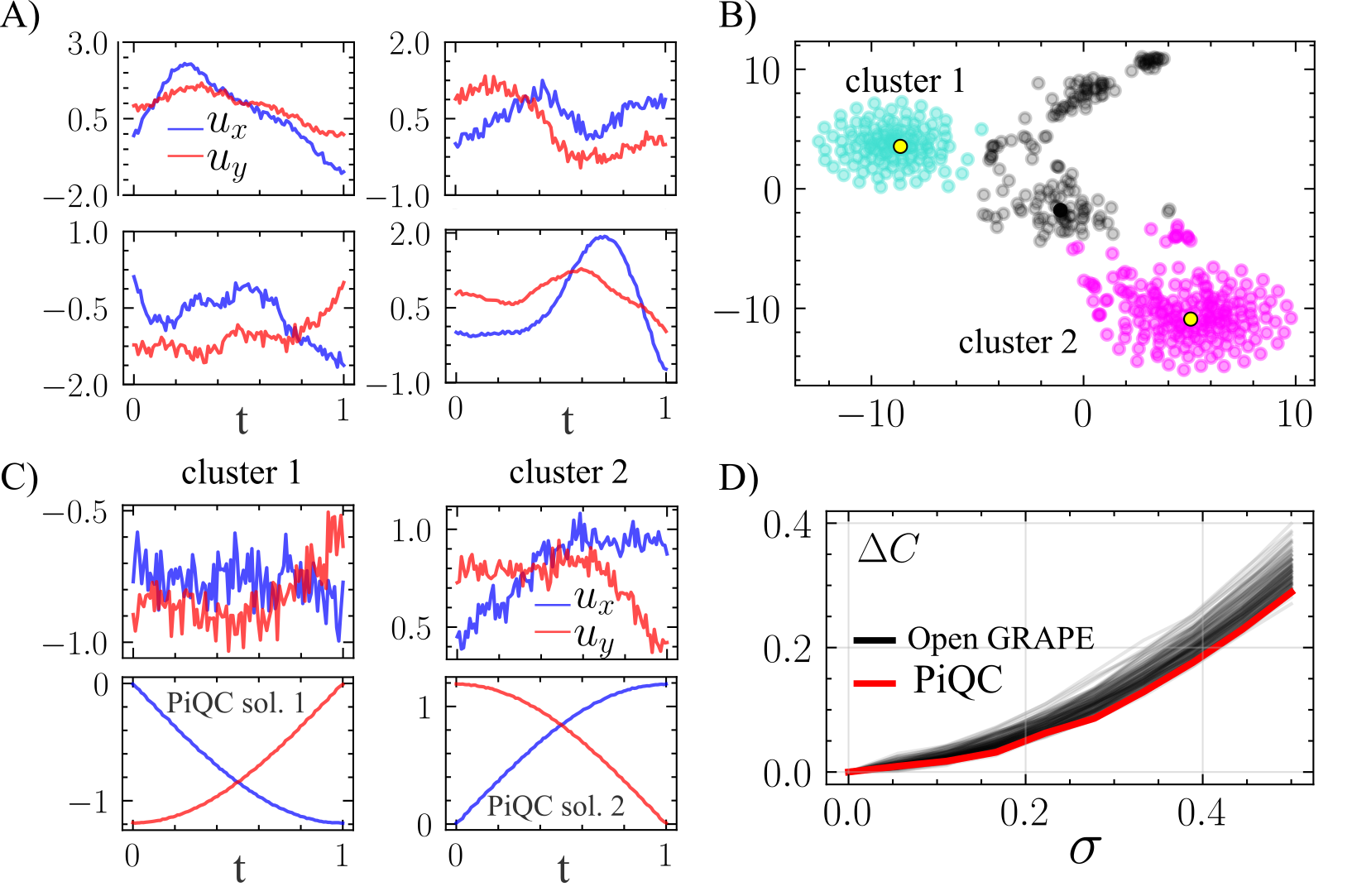}
    \end{center}
    \caption{({\bf A}) Examples of Open GRAPE solutions with diverse levels of regularity for different seeds.
    ({\bf B}) 2-D t-SNE dimension reduction representation followed by a K-means decomposition of Open GRAPE solutions (red, blue and black dots) and the two average PiQC solutions (yellow dots).
 The cluster of black dots correspond to local minima of relatively high cost $C$.
 ({\bf C}) Average Open GRAPE solutions for cluster 1 (red) and cluster 2 (blue).
 In the bottom panels, the PiQC solutions corresponding to each cluster.
 ({\bf D}) Robustness $\Delta C = C(u+\xi)-C(u)$ versus $\sigma$ with $u$ the high-fidelity high-ESS PiQC solution (red) and the high-fidelity low-ESS Open Grape solution (black).
    }
    \label{fig:tsne}
\end{figure*}

During the IS iterations of the PiQC algorithm, the optimization keeps improving in terms of ESS, even when maximum fidelities are reached; see shaded area in Fig.~\ref{fig:onequbit_XtoY} (A).
Instead, Open GRAPE solutions, while having similar fidelities as the PiQC solution, have low $ESS=\bigo{10^{-3}}$.
We expect that the high fidelity, high-ESS PiQC solutions are more robust to random perturbations than the high fidelity Open GRAPE solutions.
We test this by perturbing the optimal control solution $u(t) \to u(t) + \xi(t)$ with $\xi(t)$ with Gaussian perturbations with mean zero and standard deviation $\sigma$ and compute the maximal cost difference $\Delta C = C(u+\xi) -C(u)$ for 20 different noise realizations. 
We compare the high-fidelity Open GRAPE solutions, defined by having a cost $C < -3.8$ with the high fidelity high-ESS PiQC solution in Fig.~\ref{fig:tsne} (D) for different $\sigma$.
For completeness, we also include in the experiment the high-fidelity but low-ESS PiQC solutions corresponding to the shaded area of Fig.~\ref{fig:onequbit_XtoY} (A).
It can be observed that the PiQC solution (in red) is more robust to random perturbations than the Open GRAPE solutions (in black).
Therefore, we conclude that the ESS is a useful additional criterion to assess the quality in terms of the robustness of the (open loop) control solution against perturbations in the control fields.
See also the Supplementary Material for an experiment on robustness that suggests that PiQC is more robust under imperfect knowledge of the model Hamiltonian $H_0$.

Due to discretization artifacts and a control model consisting of a finite number of pulses, the control problem is generally non-convex, and most gradient-based algorithms suffer from local minima.
The yield of GRAPE-type algorithms, i.e. the probability of finding quality solutions, depends on the implementation and is usually modest compared to the number of trials~\cite{chen2025robust}.
For the single qubit system in this section and the NMR experiments in Sec.~\ref{sec:nmr}, the yield of PiQC can be regarded as approximately one, as we consistently reached quality solutions using only a single control seed.
One might note that the optimization process of PiQC requires sampling $\ntraj$ trajectory evolutions at each optimization step to update the control, contrary to Open GRAPE that requires only two evolutions per step, forward and backward, for updating the controls.
However, in PiQC the trajectory evolutions are highly parallelizable, as they are computed using a linear SSE that can be easily vectorized, considerably reducing the computation overhead.
The high parallelizability of PiQC speeds up the optimization and enhances the unique feature that separates the path integral control method from standard approaches, which is the adaptive importance sampling scheme.
At each optimization step, an ensemble of trajectories probes the objective landscape and collect the necessary information to update the controls for the next step.
This collective search driven by noise proves to be helpful in avoiding local minima, as we have seen in this section for the single qubit.

In this single‑qubit example, PiQC consistently converges to solutions with the same cost and fidelity from a single initialization, while GRAPE’s performance depends strongly on the initial guess---a known limitation of gradient-based local search.
We therefore compare solution quality and consistency rather than raw performance.


\subsection{Greenberger-Horne-Zeilinger (GHZ) states on open NMR systems}\label{sec:nmr}

In this section, we test the capability of our algorithm in handling systems with multiple qubits.
We demonstrate that our PiQC algorithm can successfully prepare states for realistic systems such as NMR molecules.
NMR quantum computing has gained popularity in the past years as a testbed technology for experimental and theoretical research due to the long coherence times attained compared to the dynamical scales of the system~\cite{vandersypen2004nmr}.
NMR systems equipped with transverse field controls have been demonstrated to be fully controllable for state preparation \cite{schirmer2001complete}, which initially positioned them as promising candidates for universal quantum simulators~\cite{jones1998quantum}.
We select NMR systems due to their simplicity, which enables a straightforward setup for the control problem, and allows us to compare with existing published results.
We remark that the same approach could also be applied to the coherent control of superconducting qubits~\cite{khatri2020information}, ion traps~\cite{so2024trapped, olaya2025simulating} or neutral atoms~\cite{allen2025simulating} systems.
NMR systems are defined by a total Hamiltonian $H = H_0 + H_c$, where $H_0$ is the drift and $H_c$ is the control part.
The drift and control Hamiltonians are further defined as
\bal{\label{eq:nmr_gen}
H_0 = H_Z + H_I
\qquad
H_c = \sum_{n=1}^n u_{ix} \sigma^x_i + u_{iy} \sigma^y_i
}
where $H_Z = \sum_{i=1}^n \pi D_i \sigma^z_i$ and $H_I = \sum_{i< j}^{n} \frac{\pi}{2} J_{ij}\sigma^z_i \sigma^z_j
$.
The constants $D_i$ denote self-energies, commonly called chemical shifts, and $J_{ij}$ are the coupling strengths between the different spins in the system.
Pulse-based control is a common approach to control NMR systems.
It consists in the sequential application of external radio-frequency pulses across the $x$-$y$ plane.
These pulses are represented by time-dependent fields $u_{ia}$ acting on each qubit $i$ along $x$ and $y$ directions.
 
Due to the typical large difference in magnitude between $H_Z$ and $H_I$, it is common practice to work in a rotating frame where the dynamics induced by $H_Z$ is absent.
Operationally, the rotating frame has the effect of setting $H_Z=0$ in the drift Hamiltonian $H_0$ and inducing a rotation of the controls.
The optimized pulses in the rotating frame can be easily transformed back to the laboratory frame.
We write the details about the rotation frame approach and the proof of control equivalence in the Supplementary Material.

We write the Lindblad equation for the NMR system as
\bal{\label{eq:nmr_lindblad}
\dot\rho = -i[H,\rho] +D \sum_{i=1}^n \left( \sigma_i^- \rho \sigma_i^++\sigma^+_i \rho \sigma_i^--\rho\right)
}
where $H = H_I + H_c$ and dissipation $\sigma_i^{\pm}$ of equal strength $D>0$ acting on individual spins.
The control cost is defined as
\bal{\label{eq:nmr_cost}
C= \braket{ -\frac{Q}{2} \cF(\psi_T)+ \frac{R}{2} \int_0^T \lpar \sum_{i=1}^n \sum_{a=x,y} u_{ia}^2(t) \rpar dt}
}
Here $R>0$ and satisfies $R=\lambda D^{-1}$ with $\lambda > 0$.
As in the case of one qubit, we use the transformation \cref{eq:one-qubit-invariance} to convert the non-Hermitian operators 
$C'_{ia}=\sigma_i^+, \, \sigma_i^-$ in Eq.~\cref{eq:nmr_lindblad} into anti-Hermitian operators
\bal{
C_{ib} = C'_{ia} A_{ab} \qquad \text{for all $i$}
}
This leads to the following expression for the transformed operators
\bal{\label{eq:Cxy}
C_{ib} = - i \sigma_i^b    \qquad \text{for all $i$\,\, and \,\,$b = x, y$}\,.
}
The corresponding unraveling is
\bal{\label{eq:nmr_unraveling}
d\psi =& -i H_I \psi dt -i \lpar \sum_{i=1}^n\sum_{a=x,y} (u_{ia} dt +d W_{ia} ) \sigma_i^a  \rpar \psi \notag\\
 &-n D\psi dt
}
with $dW_{ia} dW_{jb}=\delta_{ij}\delta_{ab}D dt$ for $i,j=1,\dots n$ and $a,b=x,y$.
Equations \cref{eq:nmr_cost} and \cref{eq:nmr_unraveling} define a path integral control problem.
The initial state is set to $\psi_0 = \ket{0}^n=Z^n$.
We choose the target state $\phi$ to be a Greenberger–Horne–Zeilinger state defined by $\ket{GHZ_n} := \frac{1}{\sqrt{2}} \lpar \ket{0}^n + \ket{1}^n \rpar$, which represents a maximally entangled state in the global entanglement (or Meyer-Wallach) measure for multipartite systems~\cite{meyer2002global}.
In the rotating frame, the target state takes the form $\phi = \frac{1}{\sqrt{2}} (e^{i2\pi \omega T} \ket{0}^n + e^{-i2\pi \omega T} \ket{1}^n)$ with $\omega = \frac{1}{2}\sum_j D_j$ (see Supplementary Material).

We consider the 4-qubit NMR example of \cite{chen2023accelerating}.
We add in the Supplementary Material a toy example with $n=10$ qubits to showcase the applicability of PiQC to handle larger systems.
We take the shifts $D_i$ and couplings $J_{ij}$ from Table I therein.
In~\cite{chen2023accelerating} authors use an iterative version of GRAPE, called iterative GRAPE (i-GRAPE), to optimize the pulses.
GRAPE-like methods need, in general, a high number of pulses to validate the approximation of the optimization step and reach high fidelities.
As a result, those solutions tend to be very irregular.
In addition, the variation is high between solutions for different runs, and the reported solutions for GRAPE and i-GRAPE are very different between each other.
PiQC, on the other hand, is not constrained by the number of pulses, which can be chosen arbitrarily.
Whether the solutions achieve high fidelities or not depends on other factors, such as system controllability or cost parameters, such as $R$ and $Q$.
This difference arises from the distinct approach in dealing with time correlations in PiQC compared to gradient-based methods; we return to this point in Sec.~\ref{sec:discussion}.

We couple the 4-qubit NMR molecule to an environment.
The dynamics corresponds to the Lindblad equation~\cref{eq:nmr_lindblad}.
Following~\cite{chen2023accelerating}, we choose the final time $T=8.8$ and set the noise coupling $D=\num{1e-3}/T_r$, with $T_r=0.61$ the minimum relaxation time of the system.
We compute the infidelity $1 - F_{avg}$ reached after convergence for different values of $K$.
This is shown in red in Fig.~\ref{fig:4nmr} (A).
\begin{figure*}[!t]
\begin{center}
\includegraphics[width=1\textwidth]{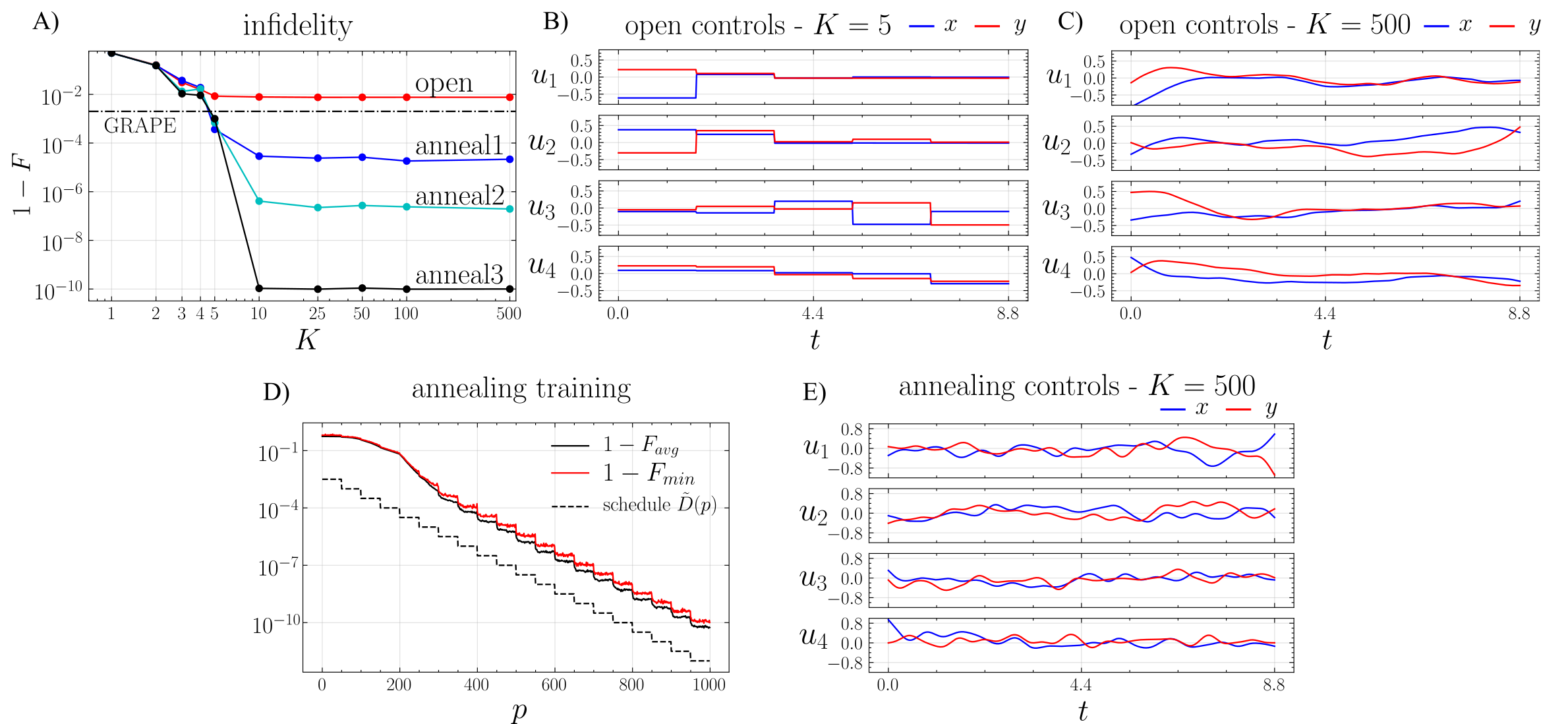}
    \end{center}
    \caption{4-NMR control. ({\bf A}) Infidelity $1-F$ vs $K$. Curve \texttt{open} corresponds to the open system case, with fixed noise coupling $D$. Curves \texttt{anneal1}, \texttt{anneal2} and \texttt{anneal3} correspond to the infidelities using annealed control. The rest of parameters are: $N_T = 500$, $n_{IS}=1000$.
We show the control solutions for the open system case for $K=5$ pulses ({\bf B}) and $K=500$ pulses ({\bf C}).
In panel ({\bf D}) we show an example of annealing training for $K=500$ pulses and logarithmic end cost.
We plot the average and the worst case infidelities as function of the IS step, and in panel ({\bf E}) we plot the corresponding final solutions.
These solutions achieve an infidelity of $\num{6e-11}$.
The annealing training is performed using an averaging window $w=10$ in the first $200$ IS steps, then setting $w=1$ for the rest of the training.
}
\label{fig:4nmr}
\end{figure*}
In the unitary case, authors in~\cite{chen2023accelerating} report solutions with infidelity $\num{2e-3}$ and $K=1760$ GRAPE pulses.
In the noisy case (non-zero $D$), we reach infidelities less than $\num{1e-2}$ with just $K=5$ pulses.
The same solutions achieve an infidelity less than $\num{8e-4}$ for the unitary case ($D=0$).
This showcases the advantage of PiQC in requiring less pulses to achieve high fidelity solutions.
In Fig.~\ref{fig:4nmr} (B-C) we show the solution profiles corresponding to $K=5$ and $K=500$ pulses, respectively.
In the $K=500$ case, the controls are iteratively smoothed during the importance sampling using cubic splines.
For large $K$, this enforces smooth solutions without sacrificing accuracy, which can be preferred in analog settings such as in neutral-atom computers or quantum annealers.

\subsection{Application of annealed diffusion control}\label{sec:numerics_annealing}

In this section we compare PiQC and GRAPE solutions in the unitary regime.
For this, we use the annealed variant of PiQC presented in Section~\ref{sec:qdcxclosed} to compute control solutions for unitary dynamics.
We apply this procedure to the 4-NMR case.
We parametrize the noise value $D$ in~\cref{eq:nmr_unraveling} in the importance sampling step $p$, interpolating between two  noise values: a high initial noise $D_0$, and a much lower noise $D_f \ll1$.
The form of the annealing schedule $D(p)$ is an ansatz that we chose empirically.
We find that good schedules allow the algorithm to sample the control landscape more effectively in early stages of the optimization.
We choose a piecewise constant schedule composed of $D_i\,(i=1, \ldots, N_{steps})$ values distributed equidistantly in the IS interval.
In the $i$-th IS interval, the noise value $D_i := D_0 \lpar \frac{D_f}{D_0} \rpar^{(i-1)/(N_{steps}-1)}$.
We set $N_{steps}=20$, $D_0 = \num{3e-3}$ and $D_f = \num{1e-12}$; see dashed line in Fig.~\ref{fig:4nmr} (D).

We compute infidelities $1-F$ vs $K$, for different values of $Q$.
In Fig.~\ref{fig:4nmr} (A) we show the infidelity curves for two values of $Q$, with fixed $R=1$.
By increasing $Q$, we expect to obtain lower infidelities, and this is shown in panel (A), where curve \texttt{anneal1} corresponds to $Q=400$, and \texttt{anneal2} to $Q=4000$.
We readily observe that PiQC solutions for these two cases produce an improvement in infidelity of $2$ and $4$ orders of magnitude, respectively, with respect to GRAPE~\cite{chen2023accelerating}.
We can push further this advantage by changing the end cost in~\cref{eq:nmr_cost} to a log-infidelity defined by $\Phi(\psi_T) = \frac{Q}{2} \log(1-\cF(\psi_T))$.
The log-infidelity is more sensitive to fine tuning when the fidelity is close to one, forcing the algorithm to keep improving the controls even when high fidelities are reached.
Applying the log-infidelity cost ($K=500$ and  $Q=400$) we observe an improvement of $7$ orders of magnitude with respect to GRAPE.
This is shown by the infidelity curve \texttt{anneal3} in panel Fig.~\ref{fig:4nmr} (A).
For this last case, we show in panel (D) the infidelity training curves (solid red and black) and the corresponding smooth solutions in panel (E).
As for the open case, the smoothing was performed during training, using cubic splines.
In this 4‑qubit NMR test case, the annealed variant of PiQC finds solutions with infidelities several orders of magnitude lower than those reported for GRAPE with comparable pulse budgets and total time.


\section{Discussions and Outlook} \label{sec:discussion}

In this work we have introduced Path integral Quantum Control (PiQC), a method for computing open-loop control solutions for both open and closed quantum systems.
PiQC is based on a stochastic reformulation of the deterministic state-preparation problem, using unravelings of the open-system dynamics to cast the control task into a path integral control form.
PiQC expresses the relevant control updates as trajectory averages, enabling direct Monte Carlo estimation on the space of stochastic quantum trajectories.
This representation allows one to approximate effective controls by sampling in Hilbert space.

A distinctive feature of PiQC is that it adopts an optimization paradigm that departs fundamentally from standard gradient-based schemes. In particular, PiQC relies on adaptive importance sampling and is therefore effectively \emph{gradient-free}: control updates are obtained by direct sampling of the cost functional, without evaluating gradients.
In generic control problems, the principal computational challenge is to assess how early control actions influence the state and cost at later times; GRAPE-like or PMP-based approaches address this via backward propagation of adjoint or co-state variables~\cite{khaneja2005optimal,Koch2019,Gautier2024AdjointOpen,Krotov1996Adjoint}.
In contrast, PiQC generates forward stochastic trajectories that are reweighted by exponentiated control costs, and uses these \emph{global statistics} to update the control parameters in each time bin independently---a hallmark of the path integral control framework.
As other trajectory-based approaches, PiQC is intrinsically \emph{state-centered}: it propagates state vectors of dimension $N$, rather than density matrices of dimension $N^2$, leading to an inherent quadratic runtime advantage over density-based approaches.
Finally, the algorithm is \emph{highly} parallel: the importance-sampling updates in Eqs.~\cref{eq:PI13}-\cref{eq:open_loop_sol} decompose into sums over trajectories, so the computation can be distributed across $m$ workers, each computing a local update $d\theta_p^{(i)}$, which are then averaged to obtain $\theta_{p+1} = \theta_p + \frac{1}{m}\sum_{i=1}^m d\theta_p^{(i)}$.
This design minimizes communication overhead, since only parameter updates---rather than raw trajectory data—need to be exchanged between machines (see Fig.~\ref{fig:algs}).
\begin{figure}[ht]
\bc
\includegraphics[width=0.4\textwidth]{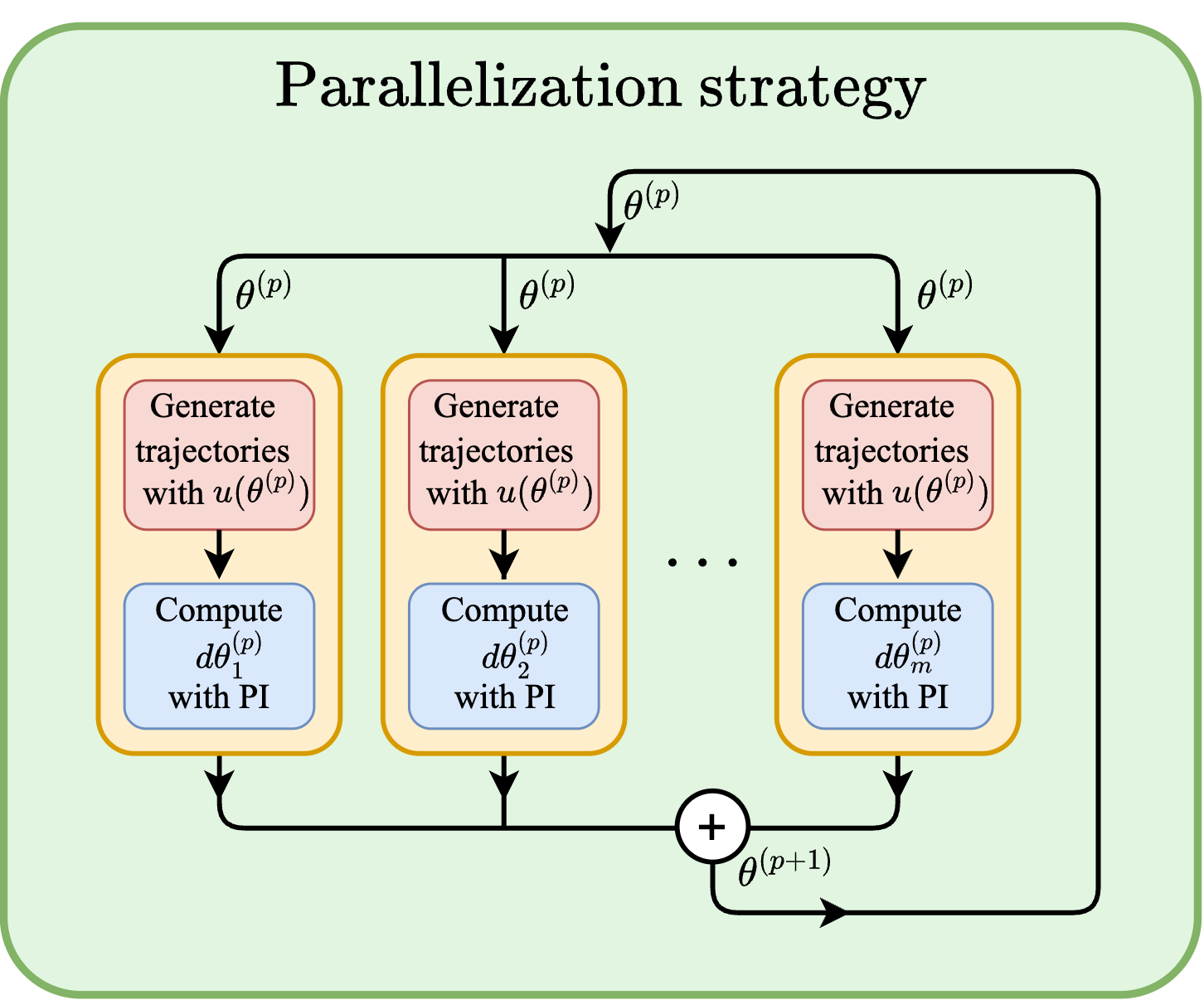}
\ec 
\caption{During each iteration, samples are generated using control functions $u(\theta^{(p)})$.
These samples are used to compute new parameters $\theta^{(p+1)}$.
These two steps can be distributed on many machines with minimal communication overhead, since no large volumes of simulated raw data need to be transferred between machines, only the parameter updates.
}
\label{fig:algs}
\end{figure}

We remark that for linear control parametrizations (an example is the pulse-based control scheme treated in this work) our method provides closed-form expressions for control updates computed as stochastic averages over quantum trajectories.
For more general non-linear parametrizations---e.g., when the control model $u(t; \theta)$ is represented by a neural network---the Path Integral Cross-Entropy (PICE) method can be employed~\cite{kappen2015adaptive}.

The main limitation of our algorithm rests in the anti-Hermitian constraint that requires the dissipators $C_a$ to live in a space spanned by anti-Hermitian operators (see Sec.~\ref{sec:stoch}).
This is a direct consequence of using a linear dynamics for the unraveling (see Eqs.~\cref{eq:control_linear} and \cref{eq:sse1}).
Despite this, there are several real world scenarios where this constraint is satisfied, as pointed out in Sec.~\cref{sec:stoch}.
For example, PiQC can be applied in dissipative settings such as in atom-light interactions~\cite{Daley2014}, relaxation and detuning in modern quantum computing platforms based on, e.g., superconducting qubits, ion-traps or neutral atoms~\cite{khatri2020information, so2024trapped, olaya2025simulating}.
For example, Ref.~\cite{keijzer2025qubit} discusses several types of noise affecting neutral-atom platforms that satisfy the anti-Hermitian constraint assumed in our framework.
The same formalism can, in principle, be used to design optimal error-mitigation strategies for coherent noise~\cite{berberich2024robustness,garcia2025resilience,shao2024multiple} in quantum computing architectures based on superconducting qubits, trapped ions, or neutral atoms operating at finite temperature~\cite{khatri2020information}. It may also be applied to mitigate Pauli-type errors, such as dephasing and depolarizing noise, and thereby to inform more resilient compilation protocols for quantum algorithms~\cite{garcia2025resilience}. Furthermore, the anti-Hermitian constraint might be relaxed by exploiting more advanced constructions, such as the iterative unraveling techniques developed in Ref.~\cite{satoh2016iterative}, which allow for more general stochastic representations beyond the PI-compatible class.
A systematic exploration of these extensions, together with a detailed analysis of how the algorithm scales with system size, is left for future work.

We further introduce an annealed variant of PiQC that employs synthetic noise as an algorithmic tool to approximate controls for closed quantum systems.
In a proof-of-principle NMR state-preparation task, this noise-assisted variant yields solutions that compare favorably with those obtained by GRAPE, and in some instances achieve substantially lower infidelities.
In this setting, the anti-Hermitian constraint is no longer tied to a physical environment, but instead enters as part of an ansatz used to compute controls in the absence of dissipation.
From this perspective, the annealed variant of PiQC can be regarded as a quantum-control heuristic analog of classical annealing methods.

An on-hardware implementation of PiQC, where quantum hardware replaces classical SSE simulation, is an interesting direction for future work, in line with existing hybrid quantum-classical optimal control schemes~\cite{Li2017HybridQOC}.
Key challenges would include injecting controllable noise into the device and accommodating hardware-native imperfections~\cite{Abbas2020NoiseResilientVQA,MitigatingControllerNoise2024,UniversallyRobustControl2025}.

In principle, the PI-control formalism underlying PiQC also applies to feedback-control settings with continuous measurements, but realizing this in practice would require incorporating measurement back-action and partial observability explicitly into the optimization.
This extension is currently under development and will be presented in future work.

\section{Acknowledgements}
We thank Christiane Koch, Luis Pedro Garcia-Pintos, Martin Larocca, Peter Komar and Roeland Wiersema for helpful discussions and advice.
We also thank Antonio Sannia for useful feedback on a previous draft of this manuscript.

\section{Data and code availability}
The data supporting the findings of this study can be fully reproduced by following the methodology described in the article.
In addition, the code necessary for replicating the PiQC results is available in a GitHub public repository~\cite{github_piqc}.

\bibliography{references}

\appendix

\section{Unravelings of the controlled Lindblad equation} \label{app:unraveling}

In the theory of stochastic quantum unravelings~\cite{Wiseman2001,hudson1984quantum,parthasarathy2012introduction,belavkin2013quantum, barchielli1991measurements, diosi1997non, gambetta2002non}, the controlled Lindblad equation~\cref{eq:lindblad} can be derived as the average dynamics over all particular time realizations of a stochastic quantum state $\psi_t$ that follows the SSE
\bal{\label{eq:sse1}
d\psi_t=& -i H_0 \psi_t dt -\frac{1}{2}D_{ab} {C_b}^\dagger C_a\psi_t dt \notag\\
&+C_a \psi_t (u_a(t) dt + dW_t^a)
}
with Ito structure $\braket{dW_t dW_t^T}=Ddt$, with $D$ a positive symmetric matrix, and $u$ an open-loop control, and $C_a = -H_a$.
The proof is straightforward using Ito calculus, and the reader may consult~\cite{percival1998quantum, barchielliQuantumTrajectoriesMeasurements2009, wiseman2009quantum, semina2014stochastic} for derivations corresponding to diagonal $D$ in the absence of control.
We give a proof for general non-diagonal $D'$ and $u\neq 0$.

Define the stochastic density operator $P_t := \psi_t\psi_t^\dagger$.
Compute the dynamics of $P_t$ as
\bal{\label{eq:sme1}
dP_t={} & d\psi_t \psi_t^\dagger + \psi_t d\psi_t^\dagger +d\psi_t d\psi_t^\dagger \notag\\
=& -i[H_0, P_t]dt +\cD[D, \cC](P_t)dt \notag\\
&+ \lbra C_a P_t (u_a(t) dt +  dW_t^a) + \hc
\rbra\,.
}
On the other hand, it is easy to show that the change in the norm $\|\psi_t\|^2 = \psi_t^\dagger \psi_t$ is given by
\baln{
d(\|\psi_t\|^2) = \psi_t^\dagger (C_a+ C_a^\dagger)\psi_t (u_a(t) dt + dW_t^a)
}
The average dynamics should correspond to a completely positive trace-preserving (CPTP) map, a general physical requirement for quantum channels~\cite{barchielliQuantumTrajectoriesMeasurements2009}.
In the stochastic setup this means that the norm of the state $\psi_t$ must preserve on average, i.e. $\braket{d\|\psi_t\|^2} =0$.
This applies for all $u$ iff $C_a$ are anti-Hermitian, i.e. $C_a = -i H_a$ for some Hermitian $H_a$.
Note that, under this condition, the norm of $\psi_t$ conserves deterministically, and all the quantum trajectories remain pure.
After inserting this into the average dynamics~\cref{eq:sme1}, taking the expectation over all noise realization, and setting $\rho(t) = \braket{P_t}$, we obtain the Lindblad equation
\baln{
\dot \rho = -i[H_0 + u_a H_a, \rho] +\cD[D, \cC](\rho)
}
Do to invariance of the dissipator $\cD[D, \cC] = \cD[D', \cC']$ under transformation~\cref{eq:invariant}, we conclude that~\cref{eq:sse1} with $C_a = -i H_a$ unravels the class of Lindblad equations with $D'$ and $\cC'$ generated by transformation~\cref{eq:invariant}.


\newpage

\onecolumngrid

\clearpage
\begin{center}
  {\Large\bfseries SUPPLEMENTARY MATERIAL}
\end{center}
\vspace{2ex}

\setcounter{section}{0} 
\renewcommand{\thesection}{S\arabic{section}} 


\setcounter{section}{0} 


\section*{Path integral control theory} \label{app:appendix_pi}

In this section we describe some basic results regarding PI control theory.
For further details the reader may consult the original papers~\cite{Kappen2005, thijssen2014a}.

Define the optimal cost-to-go as the minimal expected cost starting in any $x$ at time $t$ with $t_0\le t\le T$
\bea
J(t,x)&=&\min_{u \in \cU_t} \avg{ S^u(t)}_{X_t=x}\qquad u^*(t,x) =\underset{u \in \cU_t}{\text{argmin}}  \avg{ S^u(t)}_{X_t=x}
\eea
where $\cU_t$ is the restriction of the space of valid feedback controls $\cU$ to the time domain $[t, T]$.

Define 
\beaa
\psi(t,x)=e^{-J(t,x)/\lambda}\qquad \phi(t)=e^{(S^u(t)-S^u(t_0))/\lambda}
\eeaa
where $\lambda I=RD$. Furthermore we define the stochastic processes $\psi(t)=\psi(t,X_t),$ $u(t)=u(t,X_t)$ and $u^*(t)=u^*(t,X_t)$
\begin{lemm}
\label{lemma1}
For all $t$ in $t_0\le t\le T$ 
\bea
e^{-S^u(t)/\lambda}-\psi(t) = \frac{1}{\lambda \phi(t)}\int_t^{T} \phi(s)\psi(s) (u^*(s)-u(s))^TR dW_s\label{eq:main_lemma}
\eea
\end{lemm}
\begin{proof}
The HJB equation for the control problem Eqs.~\cref{eq:PI1} and~\cref{eq:PI2} in the paper is
\beaa
-\partial_t J=\min_u \left(V +\frac{1}{2}u^TR u + (f + g u)^T \partial_x J + \frac{1}{2}\Tr(gD g^T \partial_{xx}J)\right)
\eeaa
with the boundary condition $J(T,x)=\Phi(x)$.
We can solve for $u$, which gives $u^*=-R^{-1}g^T \partial_x J$ and
\bea
-\partial_t J=V -\frac{1}{2}\partial_x J^T g R^{-1} g^T \partial_x J + f ^T \partial_x J + \frac{1}{2}\Tr(gD g^T \partial_{xx}J)\label{eq:PI3}
\eea
We assume that the matrices $R,D$ are related such that $R=\lambda D^{-1}$ with $\lambda >0$. In terms of $\psi$, Eq.~\cref{eq:PI3} becomes linear
\bea
 \partial_t \psi + f^T\partial_x \psi +\frac{1}{2}\Tr (gD g^T \partial_{xx}\psi)=\frac{V}{\lambda}\psi \label{PI4}
\eea
Alternatively, we consider the stochastic process $\psi(t)=\psi(t,X_t)$. Using Ito calculus we obtain
\beaa
d\psi &=&\left(\partial_t \psi +\partial_x \psi^T(f+g u) +\frac{1}{2}\Tr(g D g^T \partial_{xx}\psi)\right) dt +\partial_x\psi^T g dW= \frac{V}{\lambda}\psi +\partial_x\psi^Tg (udt +dW)
\eeaa
where we used Eq.~\cref{PI4}, with boundary condition $\psi(T,x) =e^{-\Phi(x)/\lambda}$. Using the definition of $\phi(t)$ we obtain
\beaa
d\phi = -\frac{1}{\lambda} \left( V dt + u^T R dW\right)
\eeaa
with initial condition $\phi(t_0)=1$.
Using the Ito product rule we obtain
\beaa
d(\phi \psi) = d\phi \psi + \phi d\psi + d\phi d\psi=-\frac{1}{\lambda} \psi\phi u^T R dW +\phi\partial_x \psi^T g dW = \frac{1}{\lambda} \phi\psi(u^*-u)^T R dW
\eeaa
where in the last step we used $u^* = \lambda R^{-1} g^T \partial_x \psi / \psi$.
Integrating from $t$ to $T$ we obtain
\beaa
\phi(T)\psi(T)-\phi(t)\psi(t) = \frac{1}{\lambda} \int_t^{T}  \phi(s)\psi(s)(u^*(s)-u(s))^T R dW_s
\eeaa
We use $\phi(T)=\phi(t) e^{(S^u(T)-S^u(t))/\lambda}$ and divide by $\phi(t)$ to obtain Eq.~\cref{eq:main_lemma}.
\end{proof}
\begin{coro}
\label{coro1}
We take the expectation of Eq.~\cref{eq:main_lemma} with respect to $dW_{t:T}$ and condition on events up to time $t$, commonly called a {\em filtration} $\cF_t$.
We obtain 
\bea
\psi(t) = \avg{e^{-S^u(t)/\lambda}}_{\cF_t}\label{eq:PI9}
\eea
\end{coro}

The following path integral control theorem is useful to estimate these parameters for all times from one set of samples $X_{t_0:T}$. 
\begin{thrm}
\label{thm2}
Let $f: \R \times \R^n \to \R$ and define the process $f_t = f(t,X_t)$. Then
\bal{
\avg{e^{-S^u(t_0)/\lambda}\int_{t_0}^{T}(u^*(s)-u(s))f(s) ds}=\avg{e^{-S^u(t_0)/\lambda} \int_{t_0}^{T} f(s) dW_s}\label{eq:PI8}
}
where the expectation is with respect to the stochastic process \cref{eq:PI1} in the paper. 
\end{thrm}

\begin{proof}
Consider Lemma~\cref{lemma1} for $t=t_0$.
Multiply both sides by $\int_{t_0}^{T} f(s) dW_s$ and take the expectation value
\bal{
\braket{e^{-S^u(t_0)} \int_{t_0}^{T} f(s) dW_s} = \braket{\int_{t_0}^{T} f(s) dW_s \int_{t_0}^{T} \phi(s')\psi(s') (u^*(s')-u(s'))^T \frac{R}{\lambda} dW_{s'}}
}
Using $R = \lambda D^{-1}$ and Ito isometry we have
\bal{
\braket{e^{-S^u(t_0)} \int_{t_0}^{T} f(s) dW_s} =\int_{t_0}^{T} \braket{ f(s) \phi(s)\psi(s) (u^*(s)-u(s))} ds
}
Using the definition of $\psi(s)$ and $\phi(s)$ we write
\bal{
\braket{e^{-S^u(t_0)} \int_{t_0}^{T} f(s) dW_s} =\int_{t_0}^{T} \braket{ f(s) e^{(S^u(s)-S^u(t_0))/\lambda} \braket{e^{-S^u(t_0)}}_{\cF_s} (u^*(s)-u(s))} ds
}
Finally, by applying the law of total expectation we obtain \cref{eq:PI8}.

\end{proof}

\begin{coro}
\label{coro2}
By setting $f(t)=\delta_{t \in I}$ with $I=[t_0, t_0 + dt]$, diving by $dt$ and taking the limit $dt \goto 0$ we establish the proof of \cref{eq:opt_J} in the paper. 
\end{coro}


\section*{Open GRAPE} \label{app:opengrape}
Here we summarize the details of our implementation of Open GRAPE for the benchmarking part of Section~\ref{sec:noisy-qubit} in the paper.
For further details concerning the definition of Open GRAPE algorithm we refer the reader to~\cite{boutin2017opengrape}.
Given the cost objective~\cref{eq:det_cost} and the control model Eq.~\cref{eq:linear_param} in the paper, with $R$ a positive scalar, the first order Open GRAPE update rule (ignoring terms of $\cO(dt^2)$ and higher) becomes
\baln{
A_{ak} \goto A'_{ak}=A_{ak} - \epsilon \frac{\partial C}{\partial A_{ak}}
}
with
\baln{
\frac{\partial C}{\partial A_{ak}} = \frac{iQ}{2} \tr{\lambda_k[H_a, \rho_k]} dt + R A_{ak} dt
}
where $\rho_k$ is the forward-propagated state from $\rho_0$ to time $t_k$ and $\lambda_k$ is the backward-propagated state from $\rho_{target}$ to time $t_k$, $\epsilon$ is the learning rate, and $dt$ is the time step.
We implement the algorithm in Python language.
For the dynamics simulation of the backward and forward propagations we use the open-source Python package {\tt dynamiqs}~\cite{guilmin2024dynamiqs}.

The learning rate $\epsilon$ usually needs to be carefully chosen in order to avoid overshooting.
In the case of GRAPE for closed systems, work~\cite{bukov_reinforcement_2018} reported an adaptive learning rate at each $j$ step  of $\epsilon/j^\alpha$ with $\alpha=1/2$ was enough to avoid overshooting of local minima/saddle points.
We found this approach not suitable for our benchmarking purposes as it required an excessive fine tuning of $\alpha$ for different types of seeds.
We choose instead the following protocol: every time the cost $C(A') \geq C(A)$, where $A$ the current solution and $A'$ the new solution, 
the learning rate is reduced by a factor $10$ and $C(A')$ is recomputed until either $C(A')< C(A)$ or $\epsilon$ is reduced a maximum number of times that we set to $10$.
This stopping condition allows us to capture local minima and detect when the first order approximation of the gradient update starts to fail.
At the start of each run, we sample random pulse seeds $u$ from Gaussian distributions (centered at zero, 1, and -1, with standard deviation 2) and uniform distributions (integers between -10 and 10, with a scale factor of 1 and 0.5).


\section*{PiQC robustness example}\label{sec:robustness}

\begin{figure}[!t]
\bc
\includegraphics[width=0.4\textwidth]{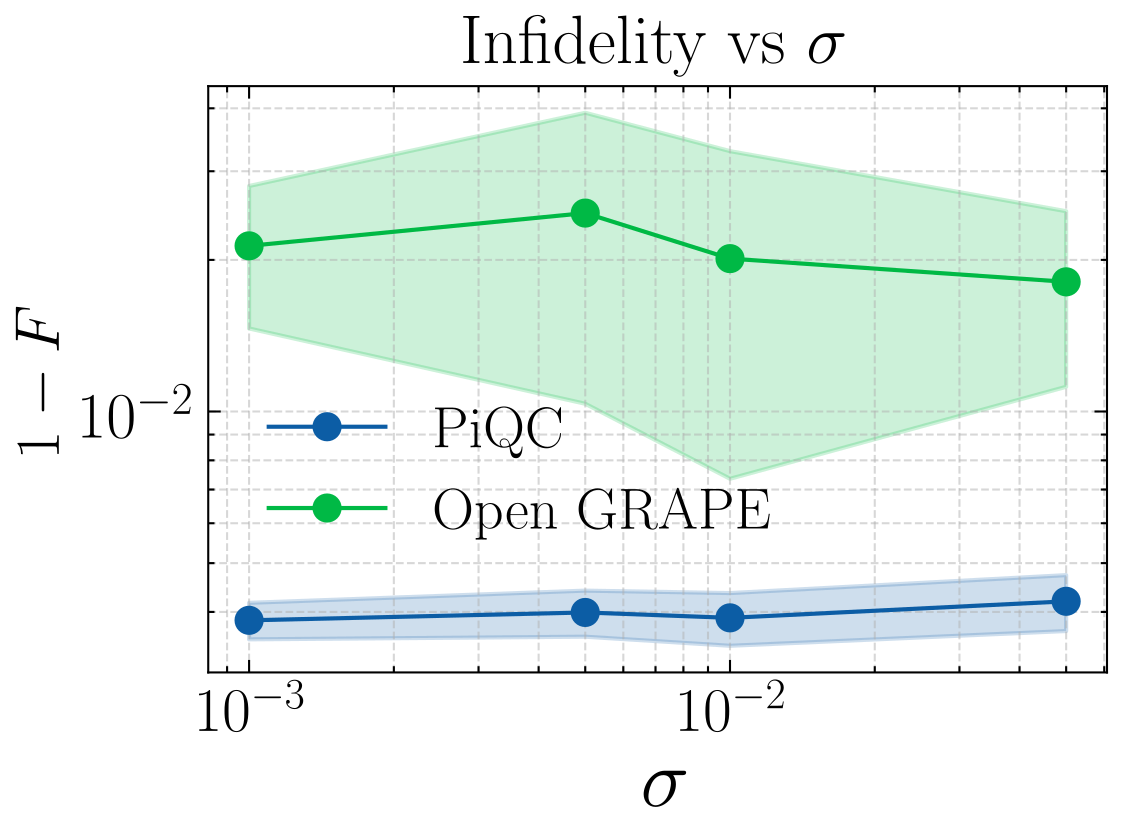}
\ec 
\caption{
Infidelity $1-F$ vs perturbation amplitude $\sigma$.
The parameters for PiQC are: $\ntraj = 400$, $n_{IS} = 500$, $R=1$, $Q=100$, $T=1$, $K=50$ and $100$ time steps.
The parameters for Open GRAPE are the same as PiQC to match the control problem definition, i.e. $R, Q, T, K$.
For each $\sigma$ we solve the control problem for each algorithm $10$ times and compute the mean and variance of the infidelities.
In the case of Open GRAPE the solutions are chosen based on best performance criteria: for each $\sigma$ we run the algorithm for 10 different control seeds and choose the one with highest fidelity.
The control seeds are normally distributed with standard deviation $5$.
The learning rate adapts according to Sec.~\ref{app:opengrape} with adaptive factor $0.5$.
The results indicate that, in this setting, PiQC attains mean infidelities that are approximately an order of magnitude lower than those obtained with Open GRAPE.
}
\label{fig:robustness}
\end{figure}

We show an experiment for the single qubit that tests PiQC robustness and algorithmic resilience under imperfect knowledge of the system model.
We consider the single qubit system of Section~\ref{sec:noisy-qubit} in the presence of dissipation, with Hamiltonian $H = H_0 + H_c(t)$, where the {\it true} drift Hamiltonian model is given by $H_0=0$ and $H_c(t)$ is the control Hamiltonian with transverse controls in $x$ and $y$.
We test the algorithm response under random perturbations of the drift Hamiltonian which encode our ignorance about the exact form of $H_0$.
For this, we consider random perturbations in time $\delta H_0(t)$ such that $H = \delta H_0(t) + H_c(t)$ such that the components of $\delta H_0(t)$ are sampled from a Gaussian distribution with a given standard deviation $\sigma$, and then taking the Hermitian part of the resulting matrix.
In consequence, the quantum trajectories will depend on the individual realizations of the Wiener increments $dW_t$ as well as $H_0(t)$.
The control problem consists of the preparation of the state $\ket{Y}$ starting from $\ket{X}$ with dissipation strength $D=0.005$.
After convergence, we use the control solutions in the true dynamics with $H_0=0$ and compute the average fidelity $F(\sigma)$ across trajectories.
We repeat this multiple times to extract some statistics and for a number of $\sigma$ values.
In parallel, we repeat the same exercise with Open Grape.
We introduce in the dynamics a time varying drift $\delta H_0(t)$ drawn from a Gaussian distribution with standard deviation $\sigma$ as before, and use the control solutions for the ground truth dynamics with $H_0=0$.
The results are shown in Fig.~\ref{fig:robustness} and demonstrate PiQC's enhanced robustness under Hamiltonian uncertainty, showing mean infidelity values one order of magnitude lower than Open GRAPE across the tested perturbation amplitudes.


\section*{Open NMR control theory} \label{app:nmr}

\subsection*{Rotating frame transformation}\label{sec:rotframe}
Define the unitary transformation
\bal{\label{eq:U}
U = \exp(i H_Z t) \qquad H_Z= \pi \sum_j D_j  \sigma^z_j
}
It represents a transformation from the lab frame to a multiple rotating frame.
Given the state $\rho$ in the lab frame, define the transformed state to the rotating frame by $\rho' = U \rho U^\dagger$.
Assume $\rho$ follows the generic Lindblad equation
\bal{
\dot \rho = -i[H_Z + H_I + H_c, \rho] + \cD[D, \cC] \rho
}
where, for simplicity, $D$ is a real scalar, $H_I$ and $H_c$ are defined as in~\cref{eq:nmr_gen}, and the $\cC = \{C_a\}_{a=1}^{n_c}$ are arbitrary dissipators.
\begin{lemm}
The transformed state $\rho'$ follows the Lindblad equation
\bal{\label{eq:nmr_rot_gen}
\dot \rho' = -i[H_I + H_c', \rho'] + \cD[D, \cC'] \rho'
}
where
\bal{\label{eq:H'c}
H_c' = \sum_i u_{ix}' \sigma^x_i + u_{iy}' \sigma^y_i \qquad C'_a = U C_a U^\dagger
}
and the controls in the rotating frame, $u'_i = (u'_{ix}, u'_{iy})^T$, are related to those in the lab frame, $u_i = (u_{ix}, u_{iy})^T$, by a real rotation
\bal{\label{eq:control_trans}
u'_i = S_i(t) u_i \qquad S_i(t) = \begin{pmatrix}
						\cos(2 \pi D_i t)  &  \sin(2 \pi D_i t)  \\
                                                -  \sin(2 \pi D_i t)  & \cos(2 \pi D_i t)
                                                \end{pmatrix}
                                                }
\end{lemm}
\begin{proof}
By differentiating $\rho'$ w.r.t. time we obtain
\baln{
\dot \rho' &= \dot U \rho U^\dagger + U \dot \rho U^\dagger + U \rho \dot U^\dagger \\
		&= -i[H_I + H_c', \rho'] + \cD[D, \cC'] \rho' \\
}

where $H_c' = U H_c U^\dagger$ and $C_a' = U C_a U^\dagger$.

Next, compute
\baln{
H'_c = U H_c U^\dagger = \sum_i u_{ix} (U \sigma^x_i U^\dagger + u_{iy} U \sigma^y_i U^\dagger)
}
For each spin $i$ we obtain
\baln{
U \sigma^x_i U^\dagger &= e^{i\pi D_i t \sigma^z_i} \sigma^x_i e^{-i\pi D_i t \sigma^z_i} \\
					&= \cos(2 \pi D_i t) \sigma^x_i - \sin(2 \pi D_i t) \sigma^y_i
}
and
\baln{
U \sigma^y_i U^\dagger &= e^{i\pi D_i t \sigma^z_i} \sigma^y_i e^{-i\pi D_i t \sigma^z_i} \\
					&= \cos(2 \pi D_i t) \sigma^y_i + \sin(2 \pi D_i t) \sigma^x_i
}
Then, we can write $H_c'$ as
\baln{
H_c' = \sum_i u_{ix}' \sigma^x_i + u_{iy}' \sigma^y_i
}
where
\bal{\label{eq:rot_control_eqs}
u_{ix}' &=  \cos(2 \pi D_i t) u_{ix} + \sin(2 \pi D_i t)  u_{iy} \nonumber \\
u_{iy}' &=  -  \sin(2 \pi D_i t)  u_{ix} + \cos(2 \pi D_i t) u_{iy} 
}
By defining 
\bal{
S_i(t) = \begin{pmatrix}
			\cos(2 \pi D_i t)  &  \sin(2 \pi D_i t)  \\
                           -  \sin(2 \pi D_i t)  & \cos(2 \pi D_i t)
                             \end{pmatrix}
}
$u'_i = (u'_{ix}, u'_{iy})^T$ and $u_i = (u_{ix}, u_{iy})^T$, we rewrite \cref{eq:rot_control_eqs} as $u'_i = S_i(t) u_i$ and the Lemma is proved.

\end{proof}

Assume the NMR system is in interaction with a electromagnetic field bath, which we encode in the dissipators $C_a = \sigma^+_i, \sigma^-_i (i=1, \ldots, n)$.

\begin{lemm}
The dissipator in \cref{eq:nmr_gen} is invariant under the action of $U$, i.e. $\cD[D, \cC'] = \cD[D, \cC]$.
\end{lemm}
\begin{proof}
Apply the unitary transformation $U$ to the Lindblad operators.
\baln{
U \sigma^+_i U^\dagger &= e^{i\pi D_i t \sigma^z_i} \sigma^+_i e^{-i\pi D_i t \sigma^z_i} \\
					&= e^{i 2\pi D_i t} \sigma^+_i
}
where we used that $[\sigma^z_i, \sigma^+_i] = 2 \sigma^+_i$.
Analogously for $\sigma^-_i$, we have $U \sigma^-_i U^\dagger = e^{-i 2\pi D_i t } \sigma^-_i$.
Therefore, the Lindblad operators $C'_a$ in the rotating frame are the same as in the lab frame up to phase factors.
As the dissipator is invariant under phase transformations, the Lemma is thus proven.
\end{proof}


\subsection*{Control problem equivalence}
The main purpose of the unitary transformation $U$ is to factor out potentially fast frequencies $D$ from the definition of the control problem that are not necessary for computing the optimal control and may even hinder the efficient search for a solution.
By defining and solving the control problem in the rotating frame, we can uniquely map the optimal solution back to the lab frame through the action of the unitary transformation in control space.
In this section we describe a general invariance present in the control problem under certain family of transformations $\cT$.

Assume the Lindblad equation for an open NMR system
\bal{\label{eq:lind}
\dot{\rho} = -i[H,\rho] + \cD[D, \cC] \rho
}
with $H = H_Z + H_I + H_c$ and dissipator
$$
\cD[D, \cC] \rho = \sum_j D_{jab} \lpar C_{ja}\rho C_{jb}^\dagger -\frac{1}{2}\{ C_{jb}^\dagger C_{ja}, \rho \} \rpar\,.
$$
where index $j$ labels qubit $j$, and the dissipation operators $C_{ja}$ act on qubit $j$.
The noise matrix $D$ is in block diagonal form $D := \bigoplus_j D_j$ with $(D_j)_{ab} = D_{jab}$.

Define the cost
\bal{\label{eq:costx}
C[x] = -\frac{Q}{2}\tr{\rho_T \rho_{tar}} + \frac{1}{2} \sum_j \int_0^T {u_j}^T R_j u_j dt 
}
where the tuple $x := (\rho_T, \rho_{tar}, u, R)$ with $R := \bigoplus_j R_j$ block diagonal.

Consider $U$ as in \cref{eq:U}, and define the transformation $\cT = (U, S)$ to the rotating frame by it's action on pair $(\rho, u)$, with $u = u_i (i=1, \ldots, n)$, as $\cT (\rho, u) = (U\rho U^\dagger, S u)$, and it's action on pair $(D, R)$ as $\cT(D, R) = (S^T D S, S^T R S)$.

Under $\cT$, the dissipators $C_{ja}$ transform as
\bal{\label{eq:Ctransf}
U C_{ja} U^\dagger = e^{i \alpha_j} \sum_b S_{jab} C_{jb}
}
where $S_j (j=1, \ldots, n)$ is an orthogonal matrix with $S_j S_j^T = 1$, $\alpha_j$ arbitrary real numbers, and define $S = \bigoplus_j S_j$.

\begin{lemm}
The dissipator $\cD[D, \cC]$ is invariant under $\cT$, i.e. $\cD[D', \cC'] = \cD[D, \cC]$.
\end{lemm}
\begin{proof}
The results follows by applying $\cT D = S D S^T$ and definition \cref{eq:Ctransf} to the dissipators $C_{ja}$.
\end{proof}

\begin{thrm} \label{th:invC}
The cost $C[x]$ in \cref{eq:costx} is invariant under $\cT$ and the optimal control solutions $u'^*$ and $u^*$ in the rotating and lab frames, respectively, are related by $u'^* = S u^*$.
\end{thrm}
\begin{proof}
Write the action of $\cT$ on $C$ as
\baln{
(\cT C)[x] &= C[\cT x] \\
		  &= -\frac{Q}{2}\tr{\rho'_T \rho'_{tar}} + \frac{1}{2} \sum_j \int_0^T {u'_j}^T R'_j u'_j dt \\
		    &=-\frac{Q}{2}\tr{U\rho_TU^\dagger U \rho_{tar}U^\dagger} + \frac{1}{2}\sum_j \int_0^T (S u_j)^T SR_j S^T (S u_j) dt \\
		       &= C[x]
}
where in the last line we used the unitarity/orthogonality of $U$ and $S_i$, respectively.
Thus, $C[x] = C[x']$ and the invariance of $C$ follows.

Next, fix $\rho_{tar}$ and $R$, and consider the optimal solution in the lab frame $u^* = \text{argmin}_u \, C[u]$, which satisfies
$
C[u^*] \leq C[u] \,\, \forall u
$
By invariance of $C$ under $\cT$, and given that $\cT$ is invertible, the transformation $S u^*$ satisfies
$
C[Su^*] \leq C[u'] \,\, \forall u'
$
where $u' = Su$.
Thus, the optimal solution in the rotating frame $u'^*$ corresponds to a rotation of the optimal solution in the lab frame, $u'^* = Su^*$.
\end{proof}

Provided that the dissipators $\cC'$ in the rotating frame can be transformed to anti-Hermitian operators, we can solve the control problem in the rotating frame using the PiQC algorithm.
By virtue of Theorem \ref{th:invC}, The optimal control $u'^*$ in the rotating frame can be mapped back to the lab frame.
Define $\cL^{rot},\, \cL^{lab}$ as the Lindbladian operators in the rotating and lab frames, respectively.
The following diagram illustrates the end-to-end workflow of the transformation.
\begin{figure}[H]
\begin{center}
\includegraphics[width=0.4\textwidth]{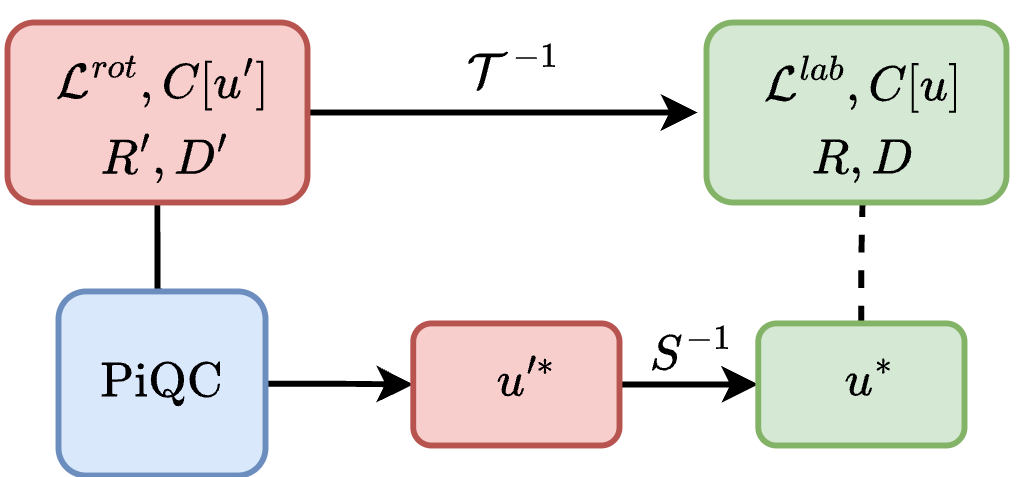}
    \end{center}
\end{figure}

\begin{example}
A particular example of $\cT$ is the case in which the dissipators $C_a = \sigma^+_i, \sigma^-_i$.
Then, $S$ corresponds to the rotation defined in \cref{eq:control_trans} and we can use the transformation $A$ defined in \cref{eq:one-qubit-invariance} to transform to anti-Hermitian operators and simulate the problem using PiQC.
\end{example}

\begin{example}
When the target state in the lab frame is a GHZ state, i.e. $\phi = \ket{GHZ}$,
the corresponding state in the rotating frame becomes $\phi' = U_T \ket{GHZ} = \frac{1}{\sqrt{2}} (e^{i2\pi \omega T} \ket{0}^n + e^{-i2\pi \omega T} \ket{1}^n)$ with $\omega = \frac{1}{2}\sum_j D_j$.
In particular, note that $\phi' = \phi$ whenever $\omega T \in \integer$.
\end{example}


\section*{A 10 qubits example} \label{sec:10qubits}
We present a complementary example applying PiQC to a 10-qubit toy model example to test the applicability of our algorithm to larger systems.

Consider a 1-D spin chain of $n$ qubits with single and pair-wise controls.
Since this is a toy model designed to test the applicability of our algorithm for a larger number of qubits, we do not assert its feasibility in a lab setting, but rather emphasize the simplicity of defining such control problems within our framework.
For simplicity, we set $H_0 = 0$.
The control Hamiltonian is
\bal{\label{eq:all_xy}
H_c = \sum_{i=1}^n \sum_{a=x,y}  u_{ia} \sigma_i^a +\sum_{i=1}^{n-1} \sum_{a,b=x,y} u_{iab} \sigma_i^a \sigma_{i+1}^b
}
where the controls $u_{ia}$ act on spin $i$, and controls $u_{iab}$ on the pair $i,\,i+1$.
Note, that $H_c$ is an extension of typical NMR control Hamiltonians~\cite{vandersypen2004nmr} with the addition of pair-wise controls.

The open dynamics is described by the Lindblad equation
\bal{\label{eq:nqubit_lindblad}
\dot{\rho}&= -i[H,\rho] +D_1 \sum_{i=1}^n \left( \sigma_i^- \rho \sigma_i^++\sigma^+_i \rho \sigma_i^--\rho\right) + D_2\sum_{i=1}^{n-1}\sum_{v,w=\pm} \left(\sigma_i^v\sigma_{i+1}^w\rho\sigma_i^{-v}\sigma_{i+1}^{-w} -\rho\right)
}
with dissipation $\sigma^{\pm}$ of equal strength $D_1$ acting on individual spins and with strength $D_2$ on neighboring pairs of spins.

The control cost is defined as
\bal{\label{eq:nqubit_cost}
C=\braket{-\frac{Q}{2}\cF(\psi_T)+ \frac{1}{2}\int_0^T \lpar R_1\sum_{i=1}^n \sum_{a=x,y}u_{ia}^2  + R_2\sum_{i=1}^{n-1}\sum_{a,b=x,y} u^2_{iab} \rpar dt}
}

Analogous to the previous section, we perform the transformation \cref{eq:one-qubit-invariance} to map non-Hermitian operators into anti-Hermitian.
We transform the operators $C'_{ia}=\sigma_i^a, a=\pm$ and $C'_{iab}=\sigma_i^a\sigma_{i+1}^b, a,b=\pm$ in Eq.~\cref{eq:nqubit_lindblad} into anti-Hermitian operators
\bal{
C_{ib} = C'_{ia} A_{ab}   \qquad C_{icd} = C'_{iab} A_{ab cd} \qquad \text{for all $i$}
}
with $A_{ab cd} := i A_{a c} A_{b d}$.
This leads to the following expression for the transformed operators
\bal{
C_{ib} = - i \sigma_i^b   \qquad C_{icd} = -i \sigma_{i}^{c} \sigma_{i+1}^{d} \qquad \text{for all $i$\,\, and \,\,$b, c, d = x, y$}\,.
}
The corresponding unraveling is
\bal{\label{eq:nqubit_unraveling}
d\psi =& -i \lpar \sum_{i=1}^n\sum_{a=x,y} (u_{ia} dt +d\tilde{W}_{ia} ) \sigma_i^a+\sum_{i=1}^{n-1} \sum_{a,b=x,y} ( u_{iab} dt +dW_{iab}) \sigma_{i}^a \sigma_{i+1}^b  \rpar \psi dt \nonumber \\
&-\left(n D_1+2(n-1) D_2\right)\psi dt
}
with $dW_{ia} dW_{jb}=\delta_{ij}\delta_{ab}D_1 dt$ for $i,j=1,\dots n$ and $a, b=x, y$, and $dW_{iab} dW_{jcd}= \delta_{ij} \delta_{ac} \delta_{bd} D_2 dt$ for $i,j= 1, \ldots, n-1$ and $a,b,c,d=x,y$.
Equations \cref{eq:nqubit_cost} and \cref{eq:nqubit_unraveling} define a path integral control problem when 
\bal{\label{eq:conditionPI2}
\lambda = D_1 R_1 = D_2 R_2
}
The initial state is $\psi_0 = \ket{0}^n$.
We choose the target state $\phi$ to be a Greenberger–Horne–Zeilinger (GHZ) state, which is defined as $\ket{GHZ_n} := \frac{1}{\sqrt{2}} \lpar \ket{0}^n + \ket{1}^n \rpar$ which represents a maximally entangled state in the global entanglement (or Meyer-Wallach) measure for multipartite systems~\cite{meyer2002global}.
We set $R_1=R_2=R/{n_c}$ with $n_c = 2n + 4(n-1)$ the total number of controls.
This normalization defines a fluence cost per unit control.
The PI condition~\cref{eq:conditionPI2} then implies that $D_2=D_1 \equiv D$.
We take as free parameters $R$ and $D_1$.
We run the algorithm for $n=10$ qubits.
The results are presented in Fig. \ref{fig:10-qubit-ghz}, highlighting the potential of PiQC for controlling large quantum systems.
\begin{figure}[!t]
    \begin{center}
        \includegraphics[width=0.95\textwidth]{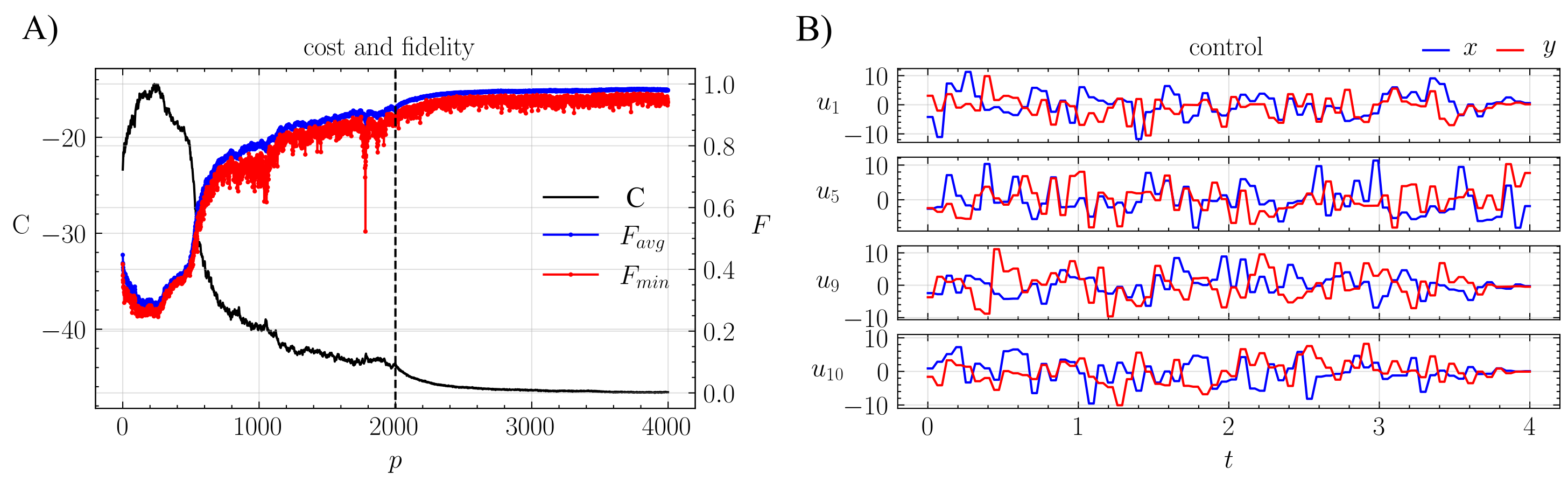}
    \end{center}
    \caption{Control of a $n=10$ qubit 1-D spin chain.
    The IS training (panel A) has two parts (marked by the vertical dashed line), where in the first part we set the IS window $w=1$ (no smoothing), and in the second part we change it to $w=10$.
This is to speedup convergence in the first part, and later to improve the statistics.
The average fidelity reached $F_{avg} = 0.982 \pm 0.001$.
Runtime $\sim 6.8$hr on a desktop computer.
In panel (B) we illustrate the control solutions for qubits 1, 5, 9 and 10.
The $x$/$y$ components of each control are marked by blue and red colors, respectively.
The parameters are $K=64$, $T=4$, $R$=0.1, $D=0.001$, $Q$=100, $n_{IS}$=4000, and $\ntraj$=100.
}
\label{fig:10-qubit-ghz}
\end{figure}
It is worth noting that a direct implementation of this problem using Open GRAPE would necessitate the computation and storage of matrix exponentials of dimensions $N^2 \times N^2$ at each step $K$, which could become prohibitive in terms of memory requirements.
This challenge underscores the need for an efficient implementation of the algorithm (see, for instance, \cite{boutin2017opengrape, abdelhafez2019quantum}) and/or the use of specialized software.
In contrast, the PiQC algorithm can run efficiently on a standard desktop computer for this number of qubits.

\end{document}